\documentclass[11pt,a4paper]{article}
\usepackage{amssymb}
\usepackage{amsmath}
\usepackage{amsfonts}
\usepackage{bbm}
\usepackage{amsthm}
\usepackage{mathrsfs}
\usepackage{hyperref}
\usepackage{color}
\usepackage[margin=2.41cm]{geometry}
\usepackage[all,cmtip]{xy}
\usepackage[utf8]{inputenc}
\usepackage{graphicx}
\usepackage{varwidth}
\usepackage{comment}

\usepackage{tocloft}
\usepackage{upgreek}
\usepackage{rotating}

\usepackage{tikz}
\usetikzlibrary{shapes.geometric}
\usetikzlibrary{decorations.markings}

\definecolor{darkred}{rgb}{0.8,0.1,0.1}
\hypersetup{
     colorlinks=false,         
     linkcolor=darkred,
     citecolor=blue,
}

\theoremstyle{plain}
\newtheorem{theo}{Theorem}[section]
\newtheorem{lem}[theo]{Lemma}
\newtheorem{propo}[theo]{Proposition}
\newtheorem{cor}[theo]{Corollary}

\theoremstyle{definition}
\newtheorem{defi}[theo]{Definition}

\newtheorem{remm}[theo]{Remark}
\newenvironment{rem}
  {\pushQED{\qed}\remm}
  {\popQED\endremm}

\numberwithin{equation}{section}

\def\bbK{\mathbb{K}}
\def\bbR{\mathbb{R}}
\def\bbC{\mathbb{C}}

\def\bbA{\mathbb{A}}
\def\bbF{\mathbb{F}}

\def\id{\mathrm{id}}

\def\dd{\mathrm{d}}
\def\vol{\mathrm{vol}}

\def\cc{\mathrm{c}}

\def\1{I}

\def\op{\mathrm{op}}

\def\Loc{\mathbf{Loc}}

\def\Set{\mathbf{Set}}
\def\Alg{\mathbf{Alg}}

\def\Vec{\mathbf{Vec}}

\def\CC{\mathbf{C}}

\def\PP{\mathbf{P}}
\def\Cat{\mathbf{Cat}}

\def\PFA{\mathbf{PFA}}
\def\toPFA{\mathbf{tPFA}}
\def\AQFT{\mathbf{AQFT}}
\def\RC{\mathbf{RC}}

\def\AAA{\mathfrak{A}}

\def\BBB{\mathfrak{B}}

\def\FFF{\mathfrak{F}}
\def\GGG{\mathfrak{G}}

\def\O{\mathcal{O}}
\def\P{\mathcal{P}}

\def\As{\mathsf{As}}

\def\colim{\mathrm{colim}}

\def\add{\mathrm{add}}
\def\addc{\mathrm{add},c}

\newcommand\und[1]{\underline{#1}}

\DeclareMathOperator*{\Motimes}{\text{\raisebox{0.25ex}{\scalebox{0.8}{$\bigotimes$}}}}

\def\sk{\vspace{2mm}}

\makeatletter
\let\@fnsymbol\@alph
\makeatother

\newcommand{\omi}[1]{\buildrel { \buildrel{#1}\over{\vee} } \over .}

\title{%
Model-independent comparison between factorization algebras \\ 
and algebraic quantum field theory on Lorentzian manifolds
}

\author{%
Marco Benini$^{1,2,a}$, 
Marco Perin$^{3,b}$\ and\
Alexander Schenkel$^{3,c}$\vspace{4mm}\\
{\small ${}^1$ Fachbereich Mathematik, Universit\"at Hamburg,}\\
{\small Bundesstr.~55, 20146 Hamburg, Germany.}\vspace{2mm}\\
{\small ${}^2$ Dipartimento di Matematica, Universit\`a di Genova,}\\
{\small Via Dodecaneso 35, 16146 Genova, Italy.}\vspace{2mm}\\
{\small ${}^3$ School of Mathematical Sciences, University of Nottingham,}\\
{\small University Park, Nottingham NG7 2RD, United Kingdom.}\vspace{4mm}\\
{\small \begin{tabular}{ll}
Email: & ${}^a$~\texttt{benini@dima.unige.it}\\
& ${}^b$~\texttt{marco.perin@nottingham.ac.uk}\\
& ${}^c$~\texttt{alexander.schenkel@nottingham.ac.uk}
\end{tabular}
}
}

\date{September 2019\vspace{-4mm}}

\begin{document}

\maketitle

\begin{abstract}
\noindent This paper investigates the relationship between algebraic quantum field theories and factorization algebras on globally hyperbolic Lorentzian manifolds. Functorial constructions that map between these two types of theories in both directions are developed under certain natural hypotheses, including suitable variants of the local constancy and descent axioms. The main result is an equivalence theorem between (Cauchy constant and additive) algebraic quantum field theories and (Cauchy constant, additive and time-orderable) prefactorization algebras. A concept of $\ast$-involution for the latter class of prefactorization algebras is introduced via transfer. This involves Cauchy constancy explicitly and does not extend to generic (time-orderable) prefactorization algebras.
\end{abstract}

\vspace{-2mm}

\paragraph*{Report no.:} ZMP-HH/19-6, Hamburger Beitr\"age zur Mathematik Nr.\ 781
\vspace{-3mm}

\vspace{-3mm}

\vspace{-1mm}

\renewcommand{\baselinestretch}{0.9}\normalsize
\tableofcontents
\renewcommand{\baselinestretch}{1.0}\normalsize


\newpage


\section{\label{sec:intro}Introduction and summary}
Factorization algebras and algebraic quantum field theory are
two mathematical frameworks to axiomatize the algebraic structure
of observables in a quantum field theory. While from a superficial
point of view these two approaches look similar, there are subtle differences.
A prefactorization algebra $\FFF$ assigns to each spacetime $M$
a {\em vector space} $\FFF(M)$ of observables
and to each tuple $\und{f} = (f_1:M_1\to N,\dots,f_n: M_n\to N)$
of pairwise disjoint spacetime embeddings a factorization product
$\FFF(\und{f}) : \bigotimes_{i=1}^n\FFF(M_i) \to \FFF(N)$ satisfying suitable
properties, cf.\ \cite{CostelloGwilliam} and Section \ref{subsec:PFA}.
On the other hand, an algebraic quantum field theory $\AAA$ assigns
to each spacetime $M$ an associative and unital {\em $\ast$-algebra}
$\AAA(M)$ of observables and to each spacetime embedding $f: M\to N$ a
$\ast$-algebra morphism $\AAA(f) : \AAA(M)\to \AAA(N)$ such that
suitable axioms hold true, cf.\ \cite{Brunetti,FewsterVerch,AQFTbook,BSWoperad}
and Section \ref{subsec:AQFT}. The main differences are that, in contrast to
an algebraic quantum field theory $\AAA$, a prefactorization
algebra $\FFF$ {\em does not} in general come endowed with 1.)~a multiplication of 
observables in $\FFF(M)$, i.e.\ on the same spacetime $M$, because $(\id_M : M\to M,\id_M:M\to M)$ 
is not a pair of disjoint spacetime embeddings, and 2.)~a concept of 
$\ast$-involution on observables  in $\FFF(M)$.
\sk

In this paper we shall develop functorial constructions (cf.\ Theorems \ref{theo:PFAtoAQFT}
and \ref{theo:AQFTtotPFA}) that allow us to relate prefactorization algebras and 
algebraic quantum field theories, provided that we assume certain natural hypotheses
on both sides. We shall focus on the case where spacetimes are described by oriented and 
time-oriented globally hyperbolic Lorentzian manifolds, i.e.\ on the case of 
relativistic quantum field theory, and disregard until Section \ref{subsec:Involutions}
the $\ast$-involutions on algebraic quantum field theories because prefactorization algebras
are usually considered without a concept of $\ast$-involution. Our main result is an equivalence theorem 
between (Cauchy constant and additive) algebraic quantum field theories
and (Cauchy constant, additive and time-orderable) prefactorization algebras,
cf.\ Theorem \ref{theo:equivalence}. Our equivalence theorem is considerably
more general than the earlier comparison result by Gwilliam and Rejzner 
\cite{GwilliamRejzner}: (1)~We work in a model-independent setup, supplemented
by natural additional hypotheses such as Cauchy constancy, additivity and 
time-orderability, while \cite{GwilliamRejzner} only studies 
linear quantum field theories, such as e.g.\ the free Klein-Gordon field. 
(2)~We investigate in detail uniqueness, associativity, naturality
and Einstein causality of the multiplications $\mu_M : \FFF(M)\otimes\FFF(M)\to \FFF(M)$ 
determined by a Cauchy constant additive prefactorization algebra $\FFF$, which requires
rather sophisticated arguments from Lorentzian geometry. These questions
were not addressed in \cite{GwilliamRejzner}. (3)~Our equivalence theorem
admits an interpretation in terms of operad theory (cf.\ Remark \ref{rem:operad}),
which provides a suitable starting point for generalizations to higher categorical
quantum field theories \cite{CostelloGwilliam,BSShocolim,BSfibered,BSWhomotopy,BSoverview} 
such as gauge theories. (The present paper does not study this generalization and
will focus on the case of $1$-categorical quantum field theories.) 
We would like to state very clearly that our results prove an equivalence theorem
between certain categories of prefactorization algebras
and algebraic quantum field theories, hence they {\em do not} 
make any statements about the relationship between explicit construction 
methods for examples. We refer to \cite{GwilliamRejzner} for 
a concrete comparison between  BV quantization \cite{CostelloGwilliam}
and perturbative canonical quantization \cite{FredenhagenRejzner, Rejzner}. 
\sk

Let us now explain in more detail our constructions and results 
while outlining the content of the present paper: In Section \ref{sec:prel}
we recall the necessary preliminaries from Lorentzian geometry, 
factorization algebras and algebraic quantum field theory. 
All prefactorization algebras and algebraic quantum field theories
will be defined on the usual category $\Loc$ of oriented and time-oriented
globally hyperbolic Lorentzian manifolds. We introduce an additivity axiom
for both prefactorization algebras and algebraic quantum field theories,
which roughly speaking demands that the observables in a spacetime $M$
are generated by the observables in the {\em relatively compact} and causally convex
open subsets $U\subseteq M$. It is shown that factorization algebras, i.e.\
prefactorization algebras satisfying Weiss descent, are in particular
additive prefactorization algebras. We also introduce a Cauchy constancy
(or time-slice) axiom for both kinds of theories, which formalizes a concept of
time evolution in a globally hyperbolic Lorentzian manifold. In Section \ref{sec:construction} we construct a functor
$\bbA: \PFA^{\addc}\to\AQFT^{\addc}$ that assigns a Cauchy constant additive algebraic
quantum field theory $\bbA[\FFF]$ to each Cauchy constant additive prefactorization algebra $\FFF$,
see Theorem \ref{theo:PFAtoAQFT} for the main result.
The crucial step is to define canonical multiplications $\mu_M : \FFF(M)\otimes \FFF(M)\to\FFF(M)$
for such $\FFF$ (cf.\ \eqref{eqn:multiplicationmap}), which is done by using
Cauchy constancy. Proving naturality and Einstein causality of these multiplications
requires the additivity axiom, cf.\ Propositions \ref{propo:multnaturality} and \ref{propo:caus}.
In Section \ref{sec:inverseconstruction} we construct
a functor $\bbF : \AQFT\to\toPFA$ that assigns a {\em time-orderable} prefactorization algebra
$\bbF[\AAA]$ to each algebraic quantum field theory $\AAA$, see Theorem \ref{theo:AQFTtotPFA} for the main result. 
The difference between time-orderable and ordinary
prefactorization algebras on $\Loc$ is that the former just encode factorization products
$\FFF(\und{f}) : \bigotimes_{i=1}^n \FFF(M_i)\to N$ for tuples of pairwise disjoint morphisms
$\und{f}$ that are in a suitable sense time-orderable, see Definition \ref{def:timeordered}.
There is a natural forgetful functor $\PFA\to\toPFA$ from ordinary to time-orderable
prefactorization algebras, which is however not full, see Remarks \ref{rem:nontimeorderable}
and \ref{rem:tPFAvsPFA}.
Our results suggest that the concept of time-orderable prefactorization algebras
from Section \ref{sec:inverseconstruction} is better suited to the 
category of Lorentzian spacetimes  $\Loc$ than the more naive concept from
Section \ref{subsec:PFA} that allows also for factorization products for non-time-orderable 
tuples of pairwise disjoint morphisms. In Section \ref{sec:equivalence} we explain
that the construction $\bbA : \PFA^{\addc}\to\AQFT^{\addc}$ from Section \ref{sec:construction}
factors through the forgetful functor $\PFA^{\addc}\to \toPFA^{\addc}$, thereby
defining a functor $\bbA: \toPFA^{\addc}\to\AQFT^{\addc}$ that assigns a Cauchy constant additive algebraic
quantum field theory to each Cauchy constant additive {\em time-orderable} prefactorization algebra.
Our main Equivalence Theorem \ref{theo:equivalence} proves that this functor
admits an inverse that is given by the restriction $\bbF : \AQFT^{\addc}\to\toPFA^{\addc}$ of the functor
from Section \ref{sec:inverseconstruction} to Cauchy constant and additive theories.
Hence, Cauchy constant additive algebraic quantum field theories are naturally identified
with Cauchy constant additive time-orderable prefactorization algebras.
In Section \ref{subsec:Involutions}, we use our main Equivalence Theorem \ref{theo:equivalence} 
to transfer $\ast$-involutions from algebraic quantum field theories to Cauchy constant additive time-orderable 
prefactorization algebras. By construction, we obtain an equivalence $\ast\AQFT^{\addc}\simeq \ast\toPFA^{\addc}$
between theories with $\ast$-involutions. We show that the transferred concept of $\ast$-involutions
for Cauchy constant additive time-orderable prefactorization algebras involves Cauchy constancy explicitly, 
hence it does not extend to generic time-orderable prefactorization algebras in $\toPFA$.
In Section \ref{subsec:KleinGordon}, we apply our general results to the simple example given by 
the free Klein-Gordon  field $\AAA_{\mathrm{KG}}\in\AQFT^{\addc}$. We observe as in \cite{GwilliamRejzner}
that the corresponding time-orderable prefactorization algebra $\FFF_{\mathrm{KG}}\in\toPFA^{\addc}$ 
describes the time-ordered products from perturbative algebraic quantum field theory, cf.\ \cite{FredenhagenRejzner, Rejzner}.


\section{\label{sec:prel}Preliminaries}
\subsection{\label{subsec:Lorentzian}Lorentzian geometry}
In order to fix our notations, we shall briefly recall some basic definitions and 
properties of Lorentzian manifolds. We refer to \cite{BGP} for a concise introduction.
\sk

A {\em Lorentzian manifold} is a manifold $M$ together with a metric $g$ of signature
$(-+\cdots+)$. A non-zero tangent vector $0\neq v\in T_x M$ at a point $x\in M$
is called {\em time-like} if $g(v,v)<0$, {\em light-like} if $g(v,v) =0$ and {\em space-like}
if $g(v,v)>0$. It is called {\em causal} 
if it is either time-like or light-like, i.e.\ $g(v,v)\leq 0$.
A curve $\gamma : I\to M$, where $I\subseteq \bbR$ is an open interval, 
is called {\em time-like}/{\em light-like}/{\em space-like}/{\em causal} 
if all its tangent vectors $\dot{\gamma}$ are time-like/light-like/space-like/causal. 
A Lorentzian manifold is called {\em time-orientable} if there exists a vector field
$\mathfrak{t}\in\Gamma^\infty(TM)$ that is everywhere time-like. 
Such $\mathfrak{t}$ determines a {\em time-orientation}.
\sk

In what follows we always consider time-oriented Lorentzian manifolds, denoted
collectively by symbols like $M$, suppressing the metric $g$ and time-orientation $\mathfrak{t}$ 
from our notation. A time-like or causal curve $\gamma : I\to M$ is called {\em future directed}
if $g(\mathfrak{t},\dot{\gamma})<0$ and {\em past directed} if $g(\mathfrak{t},\dot{\gamma})>0$.
The {\em chronological future/past} of a point $x\in M$ is the subset
$I^\pm_M(x)\subseteq M$ of all points that can be reached from $x$ by future/past directed
time-like curves. The {\em causal future/past} of a point $x\in M$ is the subset
$J^\pm_M(x)\subseteq M$ of all points that can be reached from $x$ by future/past directed
causal curves and $x$ itself. Given any subset $S\subseteq M$, we define $I^\pm_M(S) := \bigcup_{x\in S}I^\pm_M(x)$
and $J^\pm_M(S) := \bigcup_{x\in S} J^\pm_M(x)$.
\begin{defi}
Let $M$ be a time-oriented Lorentzian manifold. A subset $S\subseteq M$ is called
{\em causally convex} if $J^+_M(S)\cap J^-_M(S)\subseteq S$. Two subsets
$S,S^\prime\subseteq M$ are called {\em causally disjoint} if 
$\big(J^+_M(S)\cup J^-_M(S)\big)\cap S^\prime = \emptyset$.
\end{defi}
\begin{rem}
In words, a subset $S\subseteq M$ is causally convex if every causal curve that starts and ends
in $S$ is contained entirely in $S$. Two subsets $S,S^\prime\subseteq M$ are causally disjoint
if there exists no causal curve in $M$ connecting $S$ and $S^\prime$.
\end{rem}
\begin{defi}
A time-oriented Lorentzian manifold $M$ is called {\em globally hyperbolic}
if it admits a {\em Cauchy surface}, i.e.\ a subset $\Sigma\subset M$ 
that is met exactly once by each inextensible time-like curve in $M$.
\end{defi}

The following category of Lorentzian manifolds plays a fundamental role in 
algebraic quantum field theory, see e.g.\ \cite{Brunetti,FewsterVerch,AQFTbook,BSWoperad}.
\begin{defi}\label{def:Loc}
We denote by $\Loc$ the category whose 
objects are all oriented and time-oriented globally hyperbolic Lorentzian manifolds
$M$ and morphisms are all orientation and time-orientation
preserving isometric embeddings $f : M\to N$ with causally convex and open image $f(M)\subseteq N$.
\end{defi}

We introduce the following terminology to specify important (tuples of) $\Loc$-morphisms
that enter the definitions of algebraic quantum field theories and factorization algebras.
\begin{defi}\label{def:Locmaps}
\begin{itemize}
\item[(a)] A $\Loc$-morphism $f : M\to N$ is called a {\em Cauchy morphism}
if its image $f(M)\subseteq N$ contains a Cauchy surface of $N$. We shall write
$f : M\stackrel{c}{\to} N$ for Cauchy morphisms.

\item[(b)] A pair of $\Loc$-morphisms $(f_1:  M_1\to N,f_2: M_2\to N)$ to a common target
is called {\em causally disjoint} if the images $f_1(M_1)\subseteq N$ and $f_{2}(M_2)\subseteq N$
are causally disjoint subsets of $N$. We shall write $f_1\perp f_2$ for causally disjoint morphisms.

\item[(c)] A tuple of $\Loc$-morphisms $(f_1: M_1\to N,\dots,f_n : M_n\to N)$ to a common target
is called {\em pairwise disjoint} if the images $f_i(M_i)\subseteq N$ are pairwise disjoint
subsets of $N$, i.e. $f_i(M_i)\cap f_j(M_j) =\emptyset$, for all $i\neq j$. We shall write
$\und{f} : \und{M} \to N$ for tuples $\und{f} = (f_1,\dots,f_n)$ of pairwise disjoint morphisms.
\end{itemize}
\end{defi}
\begin{rem}
By convention, a $1$-tuple $\und{f} =(f): \und{M}\to N$ of pairwise disjoint morphisms
is just a $\Loc$-morphism $f:M\to N$ and there exists a unique  empty tuple $\emptyset\to N$
for each $N\in \Loc$.
\end{rem}

\subsection{\label{subsec:PFA}Factorization algebras}
Factorization algebras are typically considered in the context of topological, complex
or Riemannian manifolds, see \cite{CostelloGwilliam} for a detailed study. 
In order to obtain a meaningful comparison to algebraic quantum field theory, 
which is typically considered in the context of globally 
hyperbolic Lorentzian manifolds, we shall introduce below a variant of
factorization algebras on the category $\Loc$ from Definition \ref{def:Loc}.
A similar concept of factorization algebras on $\Loc$ 
appeared before in \cite{GwilliamRejzner}. For what follows let us fix any
cocomplete closed symmetric monoidal category $(\CC,\otimes,I,\tau)$,
e.g.\ the category of vector spaces $\Vec_\bbK^{}$ over a field $\bbK$.
\sk

A {\em prefactorization algebra} $\FFF$ on $\Loc$ with values in $\CC$
is given by the following data:
\begin{itemize}
\item[(i)] for each $M\in \Loc$, an object $\FFF(M)\in\CC$;
\item[(ii)] for each tuple $\und{f}=(f_1,\dots,f_n) : \und{M} \to N$
of pairwise disjoint morphisms, a $\CC$-morphism 
$\FFF(\und{f}) : \bigotimes_{i=1}^n \FFF(M_i)\to \FFF(N)$ (called {\em factorization product}),
with the convention that to the empty tuple $\emptyset \to N$ 
is assigned a morphism $I\to \FFF(N)$ from the monoidal unit.
\end{itemize}
These data are required to satisfy the following conditions:
\begin{enumerate}
\item for every $\und{f}=(f_1,\dots,f_n) : \und{M}\to N$ and $\und{g}_i =
(g_{i1},\dots,g_{i k_i}) : \und{L}_i \to M_i$, for $i=1,\dots,n$,  the diagram
\begin{flalign}\label{eqn:FAcomp}
\xymatrix@C=5em{
\ar[rd]_-{\FFF(\und{f}(\und{g}_1,\dots,\und{g}_n))~~~~~~~} \bigotimes\limits_{i=1}^n \bigotimes\limits_{j = 1}^{k_i} \FFF(L_{ij}) \ar[r]^-{\Motimes_{i} \FFF(\und{g}_i)} &
\bigotimes\limits_{i=1}^n \FFF(M_i)\ar[d]^-{\FFF(\und{f})}\\
&\FFF(N)
}
\end{flalign}
in $\CC$ commutes, where $\und{f}(\und{g}_1,\dots,\und{g}_n) :=
(f_1\,g_{11},\dots,f_n\,g_{nk_n}) : (\und{L}_1,\dots,\und{L}_n)\to N$ is given by composition in $\Loc$;

\item for every $M\in \Loc$, $\FFF(\id_M) = \id_{\FFF(M)} : \FFF(M)\to \FFF(M) $;

\item for every $\und{f}=(f_1,\dots,f_n) : \und{M}\to N$ and every permutation $\sigma\in \Sigma_n$, the diagram
\begin{flalign}\label{eqn:FAperm}
\xymatrix@C=5em{
\ar[d]_-{\text{permute}}\bigotimes\limits_{i=1}^n \FFF(M_i) \ar[r]^-{\FFF(\und{f})} & \FFF(N)\\
\bigotimes\limits_{i=1}^n \FFF(M_{\sigma(i)}) \ar[ru]_-{~~\FFF(\und{f}\sigma)}&
}
\end{flalign}
in $\CC$ commutes, where $\und{f}\sigma := (f_{\sigma(1)},\dots,f_{\sigma(n)}) : \und{M}\sigma \to N$
is given by right permutation.
\end{enumerate}
A morphism $\zeta: \FFF\to \GGG$ of prefactorization algebras 
is a family $\zeta_M : \FFF(M)\to \GGG(M)$ of $\CC$-morphisms, for all $M\in\Loc$, 
that is compatible with the factorization products, i.e.\ for all $\und{f} : \und{M}\to N$ the diagram
\begin{flalign}\label{eqn:PFACmorphism}
\xymatrix@C=5em{
\ar[d]_-{\Motimes_i \zeta_{M_i}} \bigotimes\limits_{i=1}^n \FFF(M_i) \ar[r]^-{ \FFF(\und{f})} & \FFF(N) \ar[d]^-{\zeta_N}\\
\bigotimes\limits_{i=1}^n \GGG(M_i) \ar[r]_-{ \GGG(\und{f})} & \GGG(N)
}
\end{flalign}
in $\CC$ commutes. 
\begin{defi}
We denote by $\PFA$ the category of prefactorization algebras on $\Loc$.
\end{defi}

Factorization algebras are prefactorization algebras that satisfy a suitable descent condition
with respect to Weiss covers \cite{CostelloGwilliam}. For proving our results in this
paper, it is sufficient to assume a weaker descent condition that we shall call {\em additivity}
in reference to a similar property in algebraic quantum field theory \cite{Fewster}.
As explained below, this includes in particular all factorization algebras on $\Loc$.
Before we can formalize the additivity property,  we have to introduce some further 
terminology and notations.
\begin{defi}\label{def:RCM}
For $M\in\Loc$, we denote by $\RC_M$ the {\em category of
all relatively compact and causally convex open subsets} $U\subseteq M$ 
with morphisms given by subset inclusions. 
\end{defi} 

\begin{rem}
Note that the assignment $M \mapsto \RC_M$ may be promoted to a functor
$\RC_{(-)} : \Loc\to \Cat$ with values in the category of (small) categories.
Concretely, given any $\Loc$-morphism $f:M\to N$, then the functor
$\RC_f : \RC_M\to\RC_N$ sends each relatively compact and causally convex open
subset $U\subseteq M$ to its image $f(U)\subseteq N$. Since $f$ is continuous, 
it follows that this is a relatively compact and causally convex open subset of $N$.
\end{rem}

\begin{lem}\label{lem:RCMdirected}
For every $M\in\Loc$, the category $\RC_M$ is a directed set.
\end{lem}
\begin{proof}
Let $U_1, U_2 \in \RC_M$. We shall construct $U \in \RC_M$ such that $U_i \subseteq U$, for $i=1,2$. 
Since $K := \overline{U_1} \cup \overline{U_2}$ is compact, there exists a Cauchy surface $\Sigma$ 
of $M$ such that $K \subseteq I^-_M(\Sigma)$. We set $S := J^+_M(K) \cap J^-_M(\Sigma)$ 
and observe that this is a compact subset of $M$ by \cite[Corollary~A.5.4]{BGP}. 
Using also \cite[Lemma~A.5.12]{BGP}, it follows that $U := I^+_M(K) \cap I^-_M(S)$ belongs to $\RC_M$.
By construction, $U$ contains both $U_1$ and $U_2$.
\end{proof}

We may restrict the orientation, time-orientation and metric on $M$ 
to the causally convex open subsets $U\in\RC_M$ and thereby define
objects $U\in\Loc$. Every inclusion $U\subseteq V$ in $\RC_M$
then defines a $\Loc$-morphism $\iota_U^V : U\to V$. Hence,
we can regard $\RC_M\subseteq \Loc$ as a subcategory, for
every $M\in\Loc$, and restrict any prefactorization algebra
$\FFF\in\PFA$ to a functor $\FFF\vert_{M} : \RC_M\to \CC$.
\begin{defi}\label{def:additivePFA}
A prefactorization algebra $\FFF\in\PFA$ is called {\em additive}
if, for every $M\in\Loc$, the canonical morphism
\begin{flalign}\label{eqn:additivePFA}
\xymatrix@C=3em{
\colim\Big(\FFF\vert_M : \RC_M\to \CC\Big) \ar[r]^-{\cong}~&~\FFF(M)
}
\end{flalign}
is an isomorphism in $\CC$. We denote by $\PFA^{\add}\subseteq \PFA$ the full 
subcategory of additive prefactorization algebras.
\end{defi}
\begin{rem}
The additivity condition formalizes the idea that $\FFF(M)$ is ``generated''
by the images of the maps $\FFF(U)\to\FFF(M)$, for all relatively compact and 
causally convex open subsets $U\subseteq M$. Interpreting $\FFF(M)$ as a collection
of observables for a quantum field theory, this means that all observables
described by $\FFF(M)$ arise from relatively compact regions $U\subseteq M$.
\end{rem}

\begin{propo}
Every factorization algebra $\FFF$ on $\Loc$ is an additive 
prefactorization algebra.
\end{propo}
\begin{proof}
Suppose that $\FFF$ is a factorization algebra \cite{CostelloGwilliam}, 
i.e.\ it satisfies a cosheaf condition with respect to all Weiss covers of every $M\in \Loc$. 
For every $M\in\Loc$, the cover defined by $\RC_M$ is a Weiss cover.
Indeed, given finitely many points $x_1,\dots,x_n\in M$, there exist $U_i\in\RC_M$
with $x_i\in U_i$ and hence $U\in\RC_M$  with $x_1,\dots,x_n\in U$ because
$\RC_M$ is directed by Lemma \ref{lem:RCMdirected}. The property of being a factorization
algebra then implies that the canonical diagram
\begin{flalign}\label{eqn:Cechfacalg}
\xymatrix@C=3em{
\coprod\limits_{\substack{U,V\in\RC_M \\ U\cap V\neq \emptyset}} \FFF(U\cap V) \ar@<0.5ex>[r]\ar@<-0.5ex>[r]~&~\coprod\limits_{U\in\RC_M} \FFF(U) \ar[r]~&~\FFF(M)
}
\end{flalign}
is a coequalizer in $\CC$. Our claim then follows by observing that the cocones 
of \eqref{eqn:additivePFA} are canonically identified with the cocones of 
\eqref{eqn:Cechfacalg}. Indeed, any cocone $\{\alpha_U : \FFF(U) \to Z\}$
of \eqref{eqn:additivePFA} defines a cocone of \eqref{eqn:Cechfacalg} because
$U\cap V \in\RC_M$ (whenever nonempty) and hence the diagram
\begin{subequations}
\begin{flalign}
\xymatrix@R=1em@C=3em{
&\FFF(U)\ar[dr]^-{\alpha_U}&\\
\FFF(U\cap V)\ar[rr]^-{\alpha_{U\cap V}}\ar[dr]_-{\FFF(\iota_{U\cap V}^V)~~}\ar[ur]^-{\FFF(\iota_{U\cap V}^U)~~}&&Z\\
&\FFF(V)\ar[ru]_-{\alpha_V}&
}
\end{flalign}
in $\CC$ commutes. Vice versa, any cocone $\{\alpha_U : \FFF(U) \to Z\}$ of \eqref{eqn:Cechfacalg}
defines a cocone of \eqref{eqn:additivePFA} because $U\cap V =U$, for all $U\subseteq V$,
and hence the diagram
\begin{flalign}
\xymatrix@R=1em@C=3em{
&\FFF(U)\ar[dr]^-{\alpha_U}\ar[dd]^-{\FFF(\iota_U^V)}&\\
\FFF(U\cap V)\ar[dr]_-{\FFF(\iota_{U\cap V}^V)~~}\ar@{=}[ur]&&Z\\
&\FFF(V)\ar[ru]_-{\alpha_V}&
}
\end{flalign}
\end{subequations}
in $\CC$ commutes. 
\end{proof}

As a last definition, we would like to introduce a suitable
local constancy property that is adapted to the category $\Loc$. 
This property will play a crucial role in establishing our comparison results.
Recall from Definition \ref{def:Locmaps} the concept of Cauchy morphisms.
\begin{defi}\label{def:CauchyconstantPFA}
A prefactorization algebra $\FFF\in\PFA$ is called {\em Cauchy constant}
if $\FFF(f) : \FFF(M)\to\FFF(N)$ is an isomorphism in $\CC$, 
for every Cauchy morphism $f : M\stackrel{c}{\to} N$. We denote
by $\PFA^c\subseteq \PFA$ the full subcategory of Cauchy constant prefactorization algebras.
The full subcategory $\PFA^{\addc}\subseteq \PFA^\add$ of Cauchy constant additive
prefactorization algebras is defined analogously.
\end{defi}

\subsection{\label{subsec:AQFT}Algebraic quantum field theories}
Let $\CC$ be a cocomplete closed symmetric monoidal category as in the previous subsection.
We briefly review the basic definitions for $\CC$-valued algebraic quantum field theories on $\Loc$ following
\cite{BSWoperad}. Algebraic quantum field theories with $\ast$-involutions are defined
later in Section \ref{subsec:Involutions}. We also refer to \cite{Brunetti,FewsterVerch,AQFTbook} 
for a broader introduction to algebraic quantum field theories and their applications to physics.
\sk

Let us denote by $\Alg:=\Alg_{\As}(\CC)$ the category of associative and unital algebras in $\CC$.
An {\em algebraic quantum field theory} $\AAA$ on $\Loc$ with values in $\CC$
is a functor $\AAA : \Loc\to \Alg$ that satisfies the Einstein causality axiom:
for every pair of causally disjoint morphisms $(f_1: M_1\to N)\perp (f_2:M_2\to N)$, the diagram
\begin{flalign}\label{eqn:Einsteincausality}
\xymatrix@C=5em{
\ar[d]_-{\AAA(f_1)\otimes \AAA(f_2)} \AAA(M_1)\otimes \AAA(M_2)\ar[r]^-{\AAA(f_1)\otimes \AAA(f_2)} & \AAA(N)\otimes\AAA(N)\ar[d]^-{\mu_N^\op}\\
\AAA(N)\otimes\AAA(N)\ar[r]_-{\mu_N}& \AAA(N)
}
\end{flalign}
in $\CC$ commutes, where $\mu_N^{(\op)}$ denotes the (opposite) multiplication on $\AAA(N)$.
A morphism $\kappa : \AAA\to \BBB$ of algebraic quantum field theories  is a natural transformation
between the underlying functors.
\begin{defi}
We denote by $\AQFT$ the category of algebraic quantum field theories on $\Loc$.
\end{defi}

For proving some of the results of this paper, we require 
a relatively mild variant of an additivity property in the sense of \cite{Fewster}.
Recall from Definition \ref{def:RCM} the category $\RC_M$
of relatively compact and causally convex open subsets of $M\in\Loc$.
\begin{defi}\label{def:AQFTadditive}
An algebraic quantum field theory $\AAA\in\AQFT$  is called {\em additive}
if, for every $M\in\Loc$, the canonical morphism
\begin{flalign}
\xymatrix@C=3em{
\colim\Big(\AAA\vert_M : \RC_M\to \Alg \Big) \ar[r]^-{\cong}~&~\AAA(M)
}
\end{flalign}
is an isomorphism in $\Alg$. We denote by $\AQFT^{\add}\subseteq \AQFT$ the full 
subcategory of additive algebraic quantum field theories.
\end{defi}

\begin{rem}\label{rem:algcolim}
Because $\RC_M$ is a directed set by Lemma \ref{lem:RCMdirected}, 
the colimit in Definition \ref{def:AQFTadditive}
can be computed in the underlying category $\CC$, see e.g.\ \cite[Proposition~1.3.6]{Fresse}.
Hence, to check if an algebraic quantum field theory $\AAA\in\AQFT$
is additive, one can consider its underlying
functor $\AAA : \Loc\to \CC$ to the category $\CC$ (i.e.\ forget the algebra structures)
and equivalently check if $\colim\big(\AAA\vert_M : \RC_M\to\CC\big)\to \AAA(M)$
is an isomorphism in $\CC$.
\end{rem}

Furthermore, we introduce a suitable local constancy property that 
is also known in the literature as the time-slice axiom.
\begin{defi}\label{def:AQFTCauchyconstant}
An algebraic quantum field theory $\AAA\in\AQFT$ is called
{\em Cauchy constant} if $\AAA(f) :\AAA(M)\to \AAA(N)$ is an isomorphism in $\Alg$, 
for every Cauchy morphism $f: M\stackrel{c}{\to} N$. We denote
by $\AQFT^c\subseteq \AQFT$ the full subcategory of Cauchy constant algebraic quantum field theories.
The full subcategory $\AQFT^{\addc}\subseteq \AQFT^\add$ of Cauchy constant additive
algebraic quantum field theories is defined analogously.
\end{defi}


\section{\label{sec:construction}From PFA to AQFT}
In this section we show that every Cauchy constant 
additive prefactorization algebra $\FFF\in\PFA^{\addc}$ (cf.\ Definitions 
\ref{def:additivePFA} and \ref{def:CauchyconstantPFA}) defines
a Cauchy constant additive algebraic quantum field theory
(cf.\ Definitions \ref{def:AQFTadditive} and \ref{def:AQFTCauchyconstant}). This construction
will define a functor $\bbA : \PFA^{\addc} \to \AQFT^{\addc}$.
\sk

Our construction consists of three steps, which will be carried out
in detail in individual subsections below. Step (1)~consists of proving
that, for each $M\in\Loc$, the object $\FFF(M)\in\CC$ carries canonically
the structure of an associative and unital algebra in $\CC$. This step 
relies on Cauchy constancy, while it does not require 
that the additivity property holds true.
Step (2)~consists of proving that these algebra structures are compatible
with the maps $\FFF(f) : \FFF(M)\to\FFF(N)$ induced by $\Loc$-morphisms
$f:M\to N$. Here our additivity property turns out to be crucial.
Finally, in step (3)~we show that the resulting functor $\Loc\to \Alg$
satisfies the properties of a Cauchy constant additive algebraic quantum field theory,
cf.\ Section \ref{subsec:AQFT}.
 
\subsection{Object-wise algebra structure}
All results of this subsection do not use the additivity property 
from Definition \ref{def:additivePFA}. Hence, we let $\FFF\in \PFA^{c} $ be any 
Cauchy constant prefactorization algebra.
\sk

Let us fix any $M\in\Loc$. The basic idea to define a multiplication map
$\mu_M : \FFF(M)\otimes\FFF(M)\to \FFF(M)$ is as follows: Consider
two causally convex open subsets $U_+,U_-\subseteq M$ satisfying 
(i)~there exists a Cauchy surface $\Sigma$ of $M$ such that 
$U_\pm \subseteq I^\pm_M(\Sigma)$, and (ii)~$\iota_{U_\pm}^M: U_\pm\stackrel{c}{\to} M$ 
are Cauchy morphisms. In particular, $U_+\cap U_-=\emptyset$ are disjoint
and hence provide a pair of disjoint morphisms 
$\iota_{\und{U}}^M = (\iota_{U_+}^M,\iota_{U_-}^M) : \und{U} \to M $. We define
$\mu_M$ by the commutative diagram
\begin{flalign}\label{eqn:multiplicationmap}
\xymatrix@C=3em{
\FFF(M)\otimes \FFF(M) \ar[rr]^-{\mu_M} && \FFF(M)\\
&\ar[lu]_-{\cong}^-{\FFF(\iota_{U_+}^M)\otimes \FFF(\iota_{U_-}^M)~~~~~} \FFF(U_+)\otimes \FFF(U_-)\ar[ru]_-{\FFF(\iota_{\und{U}}^M)}&
}
\end{flalign}
where the upward-left pointing arrow is an isomorphism because $\FFF$ is by hypothesis Cauchy constant.
A priori, it is not clear whether different choices of such $\iota_{\und{U}}^M : \und{U}\to M$  
lead to the same multiplication map in \eqref{eqn:multiplicationmap}. The possible choices are 
recorded in the following category.
\begin{defi}\label{def:PM}
Let $M\in\Loc$.  We denote by $\PP_M$ the category whose 
objects are all pairs of disjoint morphisms 
$\iota_{\und{U}}^M = (\iota_{U_+}^M,\iota_{U_-}^M): \und{U}\to M$
corresponding to causally convex open subsets $U_+,U_-\subseteq M$ 
that satisfy
\begin{itemize}
\item[(i)] there exists a Cauchy surface $\Sigma$ of $M$ such that $U_\pm \subseteq I^\pm_M(\Sigma)$, and
\item[(ii)] $\iota_{U_\pm}^M: U_\pm\stackrel{c}{\to} M$ are Cauchy morphisms.
\end{itemize}
There exists a unique morphism $(\iota_{\und{U}}^M :  \und{U}\to M) \to (\iota_{\und{V}}^M :  \und{V}\to M)$
if and only if $U_\pm\subseteq V_\pm$.
\end{defi}

\begin{lem}\label{lem:PPMconnected}
For every $M\in\Loc$, the category $\PP_M$ is non-empty and connected.
\end{lem}
\begin{proof}
{\em Non-empty:}  Choose any Cauchy surface $\Sigma$ of $M$ 
and define $\Sigma_\pm := I^\pm_M(\Sigma)$. Then
$\iota_{\und{\Sigma}}^M = (\iota_{\Sigma_+}^M,\iota_{\Sigma_-}^M) : \und{\Sigma}\to M$ 
defines an object in $\PP_M$.
\sk

{\em Connected:} We have to prove that there exists a zig-zag of morphisms
in $\PP_M$ between every pair of objects $\iota_{\und{U}}^M :  \und{U}\to M$ 
and $\iota_{\und{V}}^M :  \und{V}\to M$. 
For every object $\iota_{\und{U}}^M: \und{U}\to M$ in $\PP_M$,
there exists by hypothesis a Cauchy surface $\Sigma$ of $M$
such that $U_\pm \subseteq \Sigma_\pm : = I^\pm_M(\Sigma)$.
Hence, there exists a morphism $(\iota_{\und{U}}^M: \und{U}\to M)\to (\iota_{\und{\Sigma}}^M : \und{\Sigma}\to M)$.
As a consequence, our original problem reduces to finding a zig-zag
of morphisms in $\PP_M$ between $\iota_{\und{\Sigma}}^M : \und{\Sigma}\to M$
and $\iota_{\und{\Sigma}^\prime}^M : \und{\Sigma}^\prime \to M$, for any two Cauchy
surfaces $\Sigma,\Sigma^\prime$ of $M$. To exhibit such a zig-zag, let us introduce 
$\widetilde{U}_+ := \Sigma_+ \cap \Sigma^\prime_+$ and 
$\widetilde{U}_- := \Sigma_- \cap \Sigma^\prime_-$. 
If we could prove that $\iota_{\widetilde{U}_\pm}^M : \widetilde{U}_\pm \stackrel{c}{\to} M$
are Cauchy morphisms, then 
\begin{flalign}
\big(\iota_{\und{\Sigma}}^M : \und{\Sigma}\to M\big) \,\longleftarrow \,
\big(\iota_{\und{\widetilde U}}^M : \und{\widetilde{U}}\to M\big) \,\longrightarrow\,
\big(\iota_{\und{\Sigma}^\prime}^M : \und{\Sigma}^\prime\to M\big)
\end{flalign}
would provide a zig-zag that proves connectedness of $\PP_M$.
\sk

It remains to show that $\widetilde{U}_+ = \Sigma_+ \cap \Sigma^\prime_+ = 
I^+_M(\Sigma) \cap I^+_M(\Sigma^\prime)\subseteq M$ contains a Cauchy surface of $M$. (A similar argument
shows that $\widetilde{U}_-\subseteq M$ also contains a Cauchy surface of $M$.)
Because $\Sigma,\Sigma^\prime$ are by hypothesis Cauchy surfaces of $M$, there
exists a Cauchy surface $\Sigma_1\subset I^+_M(\Sigma)$ of $M$ in the future of $\Sigma$
and a Cauchy surface $\Sigma_1^\prime \subset I^+_M(\Sigma^\prime)$ of $M$ in the future of $\Sigma^\prime$.
We define the subset
\begin{flalign}\label{eqn:widetildeSigma}
\widetilde{\Sigma} \,:=\, \big(\Sigma_1 \cap J^+_M(\Sigma_1^\prime)\big) \cup \big(J^+_M(\Sigma_1)\cap \Sigma_1^\prime\big) \subset \widetilde{U}_+ \subseteq M
\end{flalign}
and claim that $\widetilde{\Sigma}$ is a Cauchy surface of $M$.
To prove the last statement, consider any inextensible time-like curve $\gamma : I\to M$,
which we may assume without loss of generality to be future directed. 
(If $\gamma$ would be past directed, then change the orientation of the interval $I$.)
Because $\Sigma_1$ and $\Sigma_1^\prime$ are Cauchy surfaces of $M$,
there exist {\em unique} $t,t^\prime\in I$ such that $\gamma(t)\in\Sigma_1$ and 
$\gamma(t^\prime)\in\Sigma_1^\prime$. If $t\geq t^\prime$, then
$\gamma(t)\in \Sigma_1 \cap J^+_M(\Sigma_1^\prime)\subseteq \widetilde{\Sigma}$,
and if $t^\prime \geq t$, then $\gamma(t^\prime) \in J^+_M(\Sigma_1)\cap 
\Sigma_1^\prime\subseteq \widetilde{\Sigma}$. Hence, $\gamma$
meets $\widetilde{\Sigma}\subset M$ at least once. Multiple intersections
are excluded by the definition of  $\widetilde{\Sigma}$ in \eqref{eqn:widetildeSigma}
and the fact that both $\Sigma_1$ and $\Sigma_1^\prime$ are Cauchy surfaces of $M$.
\end{proof}

\begin{cor}\label{cor:multunique}
For every $M\in\Loc$, the multiplication map $\mu_M$ in \eqref{eqn:multiplicationmap}
does not depend on the choice of object $\iota_{\und{U}}^M : \und{U}\to M$ in $\PP_M$.
\end{cor}
\begin{proof}
By Lemma \ref{lem:PPMconnected}, it is sufficient to prove that
$\iota_{\und{U}}^M : \und{U}\to M$ and $\iota_{\und{V}}^M : \und{V}\to M$
define the same multiplication if $U_+ \subseteq V_+$ and $U_-\subseteq V_-$. This is a consequence
of the commutative diagram
\begin{flalign}
\xymatrix@C=4em{
&\ar[dl]^-{\cong}_-{\FFF(\iota_{V_+}^M)\otimes \FFF(\iota_{V_-}^M)~~~~~}\FFF(V_+)\otimes \FFF(V_-)\ar[dr]^-{\FFF(\iota_{\und{V}}^M)} & \\
\FFF(M)\otimes \FFF(M) && \FFF(M)\\
&\ar[lu]_-{\cong}^-{\FFF(\iota_{U_+}^M)\otimes \FFF(\iota_{U_-}^M)~~~~~} \FFF(U_+)\otimes \FFF(U_-)\ar[ru]_-{\FFF(\iota_{\und{U}}^M)}\ar[uu]_-{\FFF(\iota_{U_+}^{V_+})\otimes \FFF(\iota_{U_-}^{V_-})}&
}
\end{flalign}
where one also uses the composition properties \eqref{eqn:FAcomp} 
of prefactorization algebras.
\end{proof}

To obtain a unit for $\FFF(M)$, we recall that there exists a unique empty tuple
of disjoint morphisms $\emptyset\to M$ to which the prefactorization algebra
assigns a $\CC$-morphism that we shall denote by $\eta_M : I\to \FFF(M)$.
The main result of this subsection is as follows.
\begin{propo}\label{propo:associativeunital}
Let $\FFF\in\PFA^c$ be any Cauchy constant prefactorization algebra.
For every $M\in\Loc$, the object $\FFF(M)\in\CC$ carries the structure
of an associative and unital algebra in $\CC$ with multiplication $\mu_M : \FFF(M)\otimes\FFF(M)\to \FFF(M)$ 
given by \eqref{eqn:multiplicationmap} and unit $\eta_M : I\to \FFF(M)$ given by evaluating 
$\FFF$ on the empty tuple $\emptyset \to M$.
\end{propo}
\begin{proof}
To prove that the multiplication $\mu_M$ is associative, we consider two Cauchy surfaces
$\Sigma_0,\Sigma_1$ of $M$ such that $\Sigma_1 \subset I^+_M(\Sigma_0)$, i.e.\
$\Sigma_1$ is in the future of $\Sigma_0$. Using the independence
result from Corollary \ref{cor:multunique} and the composition properties
of prefactorization algebras from Section \ref{subsec:PFA}, one easily confirms
that $\mu_M\, (\id\otimes \mu_M)$ is the upper path and $\mu_M\, (\mu_M\otimes \id )$
the lower path from $\FFF(M)^{\otimes 3}$ to $\FFF(M)$ in the commutative diagram
\begin{flalign}
\xymatrix@C=4em{
\ar[d]^-{\cong}_-{\FFF(\iota_{\Sigma_{1+}}^M)\otimes\FFF(\iota_{\Sigma_{1-}\cap\Sigma_{0+}}^M)\otimes\FFF(\iota_{\Sigma_{0-}}^M)}
\FFF(\Sigma_{1+})\otimes \FFF(\Sigma_{1-}\cap \Sigma_{0+})\otimes\FFF(\Sigma_{0-})
\ar[rr]^-{\id\otimes \FFF(\iota_{\Sigma_{1-}\cap\Sigma_{0+}}^{\Sigma_{1-}} , \iota_{\Sigma_{0-}}^{\Sigma_{1-}} )}&&
\FFF(\Sigma_{1+})\otimes \FFF(\Sigma_{1-}) \ar[d]^-{\FFF(\iota_{\Sigma_{1+}}^M , \iota_{\Sigma_{1-}}^M)}\\
\FFF(M)\otimes\FFF(M)\otimes\FFF(M) && \FFF(M)\\
\ar[u]_-{\cong}^-{\FFF(\iota_{\Sigma_{1+}}^M)\otimes\FFF(\iota_{\Sigma_{1-}\cap\Sigma_{0+}}^M)\otimes\FFF(\iota_{\Sigma_{0-}}^M)}
\FFF(\Sigma_{1+})\otimes \FFF(\Sigma_{1-}\cap \Sigma_{0+})\otimes\FFF(\Sigma_{0-}) 
\ar[rr]_-{\FFF(\iota_{\Sigma_{1+}}^{\Sigma_{0+}} , \iota_{\Sigma_{1-}\cap \Sigma_{0+}}^{\Sigma_{0+}}  ) \otimes\id }&&
\FFF(\Sigma_{0+})\otimes\FFF(\Sigma_{0-})\ar[u]_-{\FFF(\iota_{\Sigma_{0+}}^M , \iota_{\Sigma_{0-}}^M)}
}
\end{flalign}
where as before we denote by $\Sigma_{\pm } := I^\pm_M(\Sigma) \subseteq M$ the chronological
future/past of a Cauchy surface $\Sigma$ of $M$.
Unitality of the product follows immediately from the fact that there exists a unique
empty tuple  $\emptyset\to N$ for each $N\in\Loc$ and the composition properties \eqref{eqn:FAcomp}
of prefactorization algebras.
\end{proof}

\subsection{Naturality of algebra structures}
The aim of this subsection is to investigate compatibility between the algebra structures
from Proposition \ref{propo:associativeunital} and the maps $\FFF(f): \FFF(M)\to\FFF(N)$ induced by
$\Loc$-morphisms. For our main statement to be true 
it will be crucial to assume that $\FFF \in \PFA^{\addc}$
is a Cauchy constant additive prefactorization algebra
in the sense of Definitions \ref{def:additivePFA} and \ref{def:CauchyconstantPFA}.
As a first partial result, we prove the following general statement.
\begin{lem}\label{lem:RCmultcompatibility}
Let $\FFF\in\PFA^c$ be any Cauchy constant prefactorization algebra
(not necessarily additive). Let further $f:M\to N$ be a $\Loc$-morphism
such that the image $f(M)\subseteq N$ is relatively compact. Then $\FFF(f) : \FFF(M)\to\FFF(N)$
preserves the multiplications and units from Proposition \ref{propo:associativeunital}, i.e.\
$\mu_N\,(\FFF(f)\otimes\FFF(f)) = \FFF(f)\,\mu_M$ and $\eta_N = \FFF(f)\, \eta_M$.
\end{lem}
\begin{proof}
The units are clearly preserved for every $\Loc$-morphism $f:M\to N$
because composing the unique empty tuple $\emptyset \to M$ with $f:M\to N$
yields the unique empty tuple $\emptyset\to N$.
\sk

Let us focus now on the multiplications. Because $f(M)\subseteq N$ is by hypothesis
relatively compact, its closure $\overline{f(M)}\subseteq N$ is compact.
Let us take any Cauchy surface $\Sigma$ of $M$ and note that 
$\overline{f(\Sigma)} \subseteq N$ is a compact subset. Using further that $f(M)\subseteq N$
is causally convex and that the causality relation induced by time-like curves 
is open (cf.\ \cite[Lemma~14.3]{ONeill}), it follows that $\overline{f(\Sigma)} \subseteq N$ 
is achronal, i.e.\ every time-like curve in $N$ meets this subset at most once.
By \cite[Theorem~3.8]{BernalSanchez}, there exists a Cauchy surface $\widetilde{\Sigma}$ of $N$
such that $f(\Sigma) \subseteq \widetilde{\Sigma}$.
\sk

Using the Cauchy surfaces constructed above, we can define the multiplication
$\mu_{M}$ in terms of $\Sigma_\pm := I^\pm_M(\Sigma)$
and the multiplication $\mu_N $ in terms of $\widetilde{\Sigma}_\pm := I^\pm_N(\widetilde{\Sigma})$,
cf.\ \eqref{eqn:multiplicationmap}. By construction, $f:M\to N$ restricts to $\Loc$-morphisms 
$f_{\Sigma_\pm}^{\widetilde{\Sigma}_\pm} : \Sigma_\pm \to \widetilde{\Sigma}_\pm$. 
Our claim that  $\FFF(f) : \FFF(M)\to \FFF(N)$
preserves the multiplications then follows by observing that the diagram
\begin{flalign}
\xymatrix@C=7em{
\ar[d]_-{\FFF(f)\otimes\FFF(f)} \FFF(M)\otimes \FFF(M) & 
\ar[l]^-{\cong}_-{\FFF(\iota_{\Sigma_+}^M)\otimes\FFF(\iota_{\Sigma_-}^M)} \ar[d]^-{\FFF(f_{\Sigma_+}^{\widetilde{\Sigma}_+} )\otimes \FFF(f_{\Sigma_-}^{\widetilde{\Sigma}_-} )} \FFF(\Sigma_+)\otimes \FFF(\Sigma_-) \ar[r]^-{\FFF(\iota_{\und{\Sigma}}^M)}& 
\ar[d] ^-{\FFF(f)} \FFF(M)\\
\FFF(N)\otimes \FFF(N) &
 \ar[l]_-{\cong}^-{\FFF(\iota_{\widetilde{\Sigma}_+}^N) \otimes \FFF(\iota_{\widetilde{\Sigma}_-}^N) } \FFF(\widetilde{\Sigma}_+)\otimes \FFF(\widetilde{\Sigma}_-) \ar[r]_-{\FFF(\iota_{\und{\widetilde{\Sigma}}}^N)}&
\FFF(N)\\
}
\end{flalign}
commutes.
\end{proof}
\begin{rem}
We would like to emphasize that our assumption that the image $f(M)\subseteq N$ 
is relatively compact was crucial for the proof of Lemma \ref{lem:RCmultcompatibility}. 
In fact, if one {\em does not} assume that the image of the $\Loc$-morphism 
$f: M\to N$ is relatively compact, then it is {\em not} true that the image $f(\Sigma)\subset N$
of a Cauchy surface $\Sigma$ of $M$ can be extended to a Cauchy surface $\widetilde{\Sigma}$
of $N$. A simple example that demonstrates this feature is given by the subset inclusion
$\iota_{U}^V : U\to V$ of the following two diamond regions in $2$-dimensional 
Minkowski spacetime (note that $U$ is not relatively compact as a subset of $V$):
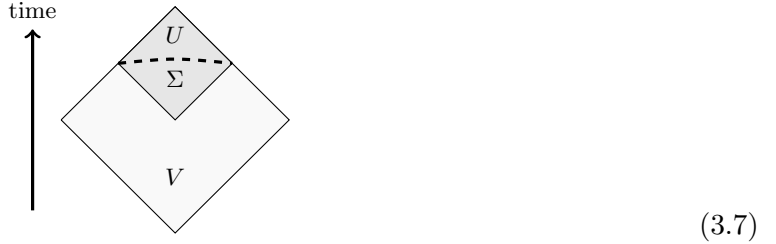
\begin{flalign}\label{eqn:nonrelativelycompact}
\begin{tikzpicture}[scale=1.5]
\draw[fill=gray!5] (-1,0) -- (0,1) -- (1,0) -- (0,-1) -- (-1,0);
\draw[fill=gray!20] (-0.5,0.5) -- (0,1) -- (0.5,0.5) -- (0,0) -- (-0.5,0.5);
\draw (0,-0.5) node {{\footnotesize $V$}};
\draw (0,0.75) node {{\footnotesize $U$}}; 
\draw[very thick, dashed] (-0.5,0.5) .. controls (0,0.55) .. (0.5, 0.5) node[midway, below] {{\footnotesize $\Sigma$}};
\draw[very thick, ->] (-1.25,-0.8) -- (-1.25,0.8) node[above] {{\footnotesize time}};
\end{tikzpicture}
\end{flalign}
It is evident that no Cauchy surface $\Sigma$ of $U$ admits an extension to a Cauchy surface of $V$.
Hence, $\FFF(\iota_{U}^V) : \FFF(U)\to \FFF(V)$ may fail to preserve the multiplications.
We shall show below that the issues explained in this remark are solved by 
considering {\em additive} prefactorization algebras as in Definition \ref{def:additivePFA}.
\end{rem}

The main result of this subsection is as follows.
\begin{propo}\label{propo:multnaturality}
Let $\FFF\in\PFA^{\addc}$ be any Cauchy constant additive 
prefactorization algebra. For every $\Loc$-morphism $f:M\to N$,
the $\CC$-morphism $\FFF(f) : \FFF(M)\to\FFF(N)$ preserves the multiplications and units 
from Proposition \ref{propo:associativeunital}. 
\end{propo}
\begin{proof}
We already observed in the proof of Lemma \ref{lem:RCmultcompatibility}
that $\FFF(f) $ preserves the units. 
\sk

For the multiplications we have to prove that $\mu_N\,(\FFF(f)\otimes\FFF(f)) = \FFF(f)\,\mu_M$
as $\CC$-morphisms from $\FFF(M)\otimes\FFF(M)$ to $\FFF(N)$. Because $\FFF$ is by hypothesis
additive (cf.\ Definition \ref{def:additivePFA}) and the monoidal product $\otimes$
in a cocomplete closed symmetric monoidal category preserves colimits in both entries,
it follows that 
\begin{flalign}\label{eqn:colimtmpx}
\FFF(M)\otimes\FFF(M) \,\cong\, \colim_{U,V\in \RC_M}^{~} \big(\FFF(U)\otimes \FFF(V)\big)\,\cong\,
\colim_{U\in\RC_M}^{~} \big(\FFF(U)\otimes \FFF(U)\big)\quad,
\end{flalign}
where in the last step we also used that $\RC_M$ is directed by Lemma \ref{lem:RCMdirected}.
For every $U\in\RC_M$, consider the diagram
\begin{flalign}\label{eqn:diagramformultcompatible}
\xymatrix@C=5em{
\ar[ddr]_-{\FFF(f_U^{~})\otimes \FFF(f_U^{~})~} \FFF(U)\otimes\FFF(U) \ar[rd]^(.7){~~~\FFF(\iota_U^M)\otimes\FFF(\iota_U^M)}\ar[r]^-{\mu_U}& \FFF(U) \ar[rd]^-{\FFF(\iota_U^M)} \ar[rdd]_(.7){\FFF(f_U^{~})}|\hole&\\
&\ar[d]^-{\FFF(f)\otimes\FFF(f)}\FFF(M)\otimes \FFF(M)\ar[r]^-{\mu_M}& \FFF(M) \ar[d]^-{\FFF(f)}\\
&\FFF(N)\otimes\FFF(N)\ar[r]_-{\mu_N}& \FFF(N)
}
\end{flalign}
where $f_U^{~} : U\to N$ denotes the restriction of $f:M\to N$
to $U\subseteq M$.
The top and bottom squares of this diagram commute because
of Lemma \ref{lem:RCmultcompatibility} and the fact that both $U\subseteq M$ 
and $f(U)\subseteq N$ are relatively compact subsets. 
The two triangles commute by direct inspection. By universality of the colimit in \eqref{eqn:colimtmpx}, 
this implies that the front square in \eqref{eqn:diagramformultcompatible} commutes, 
proving our claim. 
\end{proof}

\begin{cor}\label{cor:natalg}
Every Cauchy constant additive prefactorization algebra
$\FFF \in\PFA^{\addc}$ defines a functor
$\bbA[\FFF] : \Loc\to \Alg$ to the category of associative and unital algebras.
Explicitly, this functor acts on objects $M\in\Loc$ as $\bbA[\FFF](M) := (\FFF(M),\mu_M,\eta_M)$
and on $\Loc$-morphisms $f:M\to N$ as $\bbA[\FFF](f) :=\FFF(f)$.
The assignment $\FFF \mapsto \bbA[\FFF]$ canonically extends to a functor
$\bbA : \PFA^{\addc} \to \Alg^\Loc$, where $\Alg^\Loc$ denotes the category of functors
from $\Loc$ to $\Alg$.
\end{cor}
\begin{proof}
It remains to prove that every morphism $\zeta : \FFF\to\GGG$ in $\PFA^{\addc}$ defines
a natural transformation $\bbA[\zeta]: \bbA[\FFF]\to \bbA[\GGG]$ between $\Alg$-valued functors on $\Loc$,
i.e.\ that all components $\zeta_M : \FFF(M)\to \GGG(M)$ preserve the multiplications and units.
For the units this is immediate, while for the multiplications it follows from the fact that the diagram
\begin{flalign}
\xymatrix@C=7em{
\ar[d]_-{\zeta_M\otimes\zeta_M} \FFF(M)\otimes\FFF(M) & \ar[l]^-{\cong}_-{\FFF(\iota_{U_+}^M)\otimes\FFF(\iota_{U_-}^M)} \ar[d]_-{\zeta_{U_+}\otimes \zeta_{U_-}}\FFF(U_+) \otimes \FFF(U_-) \ar[r]^-{\FFF(\iota_{\und{U}}^M)}& \FFF(M)\ar[d]^-{\zeta_M}\\
\GGG(M)\otimes\GGG(M) &  \ar[l]_-{\cong}^-{\GGG(\iota_{U_+}^M)\otimes\GGG(\iota_{U_-}^M)}\GGG(U_+) \otimes \GGG(U_-)\ar[r]_-{\GGG(\iota_{\und{U}}^M)} & \GGG(M)
}
\end{flalign}
commutes by the compatibility properties \eqref{eqn:PFACmorphism} 
of prefactorization algebra morphisms.
\end{proof}

\subsection{Algebraic quantum field theory axioms}
The goal of this subsection is to show that the construction above 
assigns to each Cauchy constant additive prefactorization algebra 
a Cauchy constant additive algebraic quantum field theory. More precisely, we shall prove that 
the functor $\bbA : \PFA^{\addc} \to \Alg^\Loc$ 
established in Corollary \ref{cor:natalg} factors through 
the full subcategory $\AQFT^{\addc} \subseteq \Alg^\Loc$ 
of Cauchy constant additive algebraic quantum field theories. 

\begin{lem}\label{lem:caus}
Let $\FFF \in \PFA^c$ be any Cauchy constant prefactorization algebra (not necessarily additive). 
Let further $(f_1: M_1 \to N)\perp (f_2: M_2 \to N)$ be any causally disjoint pair of $\Loc$-morphisms 
such that the images $f_1(M_1), f_2(M_2) \subseteq N$ are relatively compact. 
Then $\mu_N^\op\, (\FFF(f_1) \otimes \FFF(f_2)) = \mu_N\, (\FFF(f_1) \otimes \FFF(f_2))$, 
where $\mu_N^{(\op)}$ denotes the (opposite) multiplication on $\FFF(N)$ from 
Proposition \ref{propo:associativeunital}. 
\end{lem}
\begin{proof}
In order to compare the two morphisms $\mu_N\, (\FFF(f_1) \otimes \FFF(f_2))$ 
and $\mu_N^{\op}\, (\FFF(f_1) \otimes \FFF(f_2))$ from $\FFF(M_1) \otimes \FFF(M_2)$ to $\FFF(N)$, 
we introduce convenient ways to compute these composites. 
Let us choose arbitrary Cauchy surfaces $\Sigma_1$ of $M_1$ and $\Sigma_2$ of $M_2$. 
As in the proof of Lemma \ref{lem:RCmultcompatibility}, we deduce that 
$\overline{f_1(\Sigma_1)}, \overline{f_2(\Sigma_2)} \subseteq N$ are achronal compact subsets. 
Causal disjointness of the pair $f_1\perp f_2$
entails achronality of the union $\overline{f_1(\Sigma_1)} \cup \overline{f_2(\Sigma_2)} \subseteq N$. 
By \cite[Theorem~3.8]{BernalSanchez}, there exists a Cauchy surface $\widetilde{\Sigma}$ of $N$ 
that contains the union $\overline{f_1(\Sigma_1)} \cup \overline{f_2(\Sigma_2)}\subseteq \widetilde{\Sigma}$. 
Similarly, choosing any Cauchy surface $\Sigma^\prime_1 \subset I_{M_1}^+(\Sigma_1)$ of $M_1$ 
that lies in the future of $\Sigma_1$ and any Cauchy surface $\Sigma^\prime_2 \subset I_{M_2}^-(\Sigma_2)$ 
of $M_2$ that lies in the past of $\Sigma_2$, 
there exists a Cauchy surface $\widetilde{\Sigma}^\prime$ of $N$ that contains the union 
$\overline{f_1(\Sigma^\prime_1)} \cup \overline{f_2(\Sigma^\prime_2)}\subseteq \widetilde{\Sigma}^\prime$. 
Let us introduce 
\begin{flalign}
U_1 \,:=\, I_{M_1}^+(\Sigma_1) \cap I_{M_1}^-(\Sigma^\prime_1) \subseteq M_1\quad,\qquad
U_2 \, :=\, I_{M_2}^+(\Sigma^\prime_2) \cap I_{M_2}^-(\Sigma_2) \subseteq M_2\quad,
\end{flalign}
and also consider $\widetilde{\Sigma}_\pm := I_N^\pm(\widetilde{\Sigma}) \subseteq N$ 
and $\widetilde{\Sigma}^\prime_\pm := I_N^\pm(\widetilde{\Sigma}^\prime) \subseteq N$. 
By construction, $\iota_{U_i}^{M_i}: U_i \stackrel{c}{\to} M_i$, for $i=1,2$, and 
$\iota_{\widetilde{\Sigma}^{(\prime)}_\pm}^N: \widetilde{\Sigma}^{(\prime)}_\pm \stackrel{c}{\to} N$ 
are Cauchy morphisms. The following picture illustrates in dark gray the
chosen subsets $U_1\subseteq M_1$ and $U_2\subseteq M_2$: 
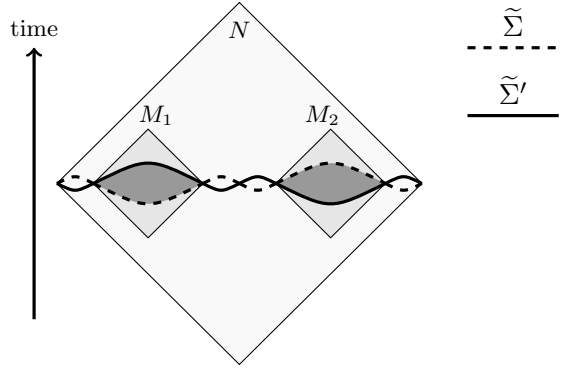
\begin{flalign}
\begin{tikzpicture}[scale=1.2]
\draw[fill=gray!5] (-2,0) -- (0,2) -- (2,0) -- (0,-2) -- (-2,0);
\draw (0,1.7) node {{\footnotesize $N$}};
\draw [fill=gray!20] (-1.5,0) -- (-1,0.5) -- (-0.5,0) -- (-1,-0.5) -- (-1.5,0);
\draw [fill=gray!20] (-1.6,0) -- (-1,0.6) -- (-0.4,0) -- (-1,-0.6) -- (-1.6,0);
\draw [fill=gray!20] (1.6,0) -- (1,0.6) -- (0.4,0) -- (1,-0.6) -- (1.6,0);
\draw (-1,0.75) node {{\footnotesize $~~M_1$}};
\draw (1,0.75) node {{\footnotesize $M_2~~$}};
\draw[fill=gray!80,draw=none] (-1.6,0) .. controls (-1,0.3) .. (-0.4,0) .. controls (-1,-0.3) .. (-1.6,0);
\draw[fill=gray!80,draw=none] (1.6,0) .. controls (1,0.3) .. (0.4,0) .. controls (1,-0.3) .. (1.6,0);
\draw[very thick] (-2,0) .. controls (-1.8,-0.1) .. (-1.6,0);
\draw[very thick,dashed] (-2,0) .. controls (-1.8,0.1) .. (-1.6,0);
\draw[very thick] (-1.6,0) .. controls (-1,0.3) .. (-0.4,0);
\draw[very thick, dashed] (-1.6,0) .. controls (-1,-0.3) .. (-0.4,0);
\draw[very thick] (-0.4,0) .. controls (-0.2,-0.1) .. (0,0);
\draw[very thick] (0,0) .. controls (0.2,0.1) .. (0.4,0);
\draw[very thick, dashed] (-0.4,0) .. controls (-0.2,0.1) .. (0,0);
\draw[very thick, dashed] (0,0) .. controls (0.2,-0.1) .. (0.4,0);
\draw[very thick] (0.4,0) .. controls (1,-0.3) .. (1.6,0);
\draw[very thick, dashed] (0.4,0) .. controls (1,0.3) .. (1.6,0);
\draw[very thick] (1.6,0) .. controls (1.8,0.1) .. (2,0);
\draw[very thick, dashed] (1.6,0) .. controls (1.8,-0.1) .. (2,0);
\draw[very thick, dashed] (2.5,1.5) -- (3.5,1.5) node[midway,above] {{$\widetilde{\Sigma}$}};
\draw[very thick] (2.5,0.75) -- (3.5,0.75) node[midway,above] {{$\widetilde{\Sigma}^\prime$}};
\draw[very thick, ->] (-2.25,-1.5) -- (-2.25,1.5) node[above] {{\footnotesize time}};
\end{tikzpicture}
\end{flalign}
With these preparations, we can compute $\mu_N\, (\FFF(f_1) \otimes \FFF(f_2))$ by
\begin{flalign}
\xymatrix@C=8em{
\FFF(M_1)\otimes\FFF(M_2) \ar[r]^-{\FFF(f_1) \otimes \FFF(f_2)} &
\FFF(N)\otimes\FFF(N) \ar[r]^-{\mu_N} & \FFF(N) \\
\ar[u]_-{\cong}^-{\FFF(\iota_{U_1}^{M_1}) \otimes \FFF(\iota_{U_2}^{M_2})} \FFF(U_1) \otimes \FFF(U_2) \ar[r]_-{\FFF\big((f_1)_{U_1}^{\widetilde{\Sigma}_+}\big) \otimes \FFF\big((f_2)_{U_2}^{\widetilde{\Sigma}_-}\big)} 
& \ar[u]_-{\cong}^-{\FFF(\iota_{\widetilde{\Sigma}_+}^N) \otimes \FFF(\iota_{\widetilde{\Sigma}_-}^N)} \FFF(\widetilde{\Sigma}_+) \otimes \FFF(\widetilde{\Sigma}_-) \ar[ur]_-{\FFF(\iota_{\und{\widetilde{\Sigma}}}^N)} 
}
\end{flalign}
where $(f_1)_{U_1}^{\widetilde{\Sigma}_+} : U_{1} \to \widetilde{\Sigma}_{+}$ denotes 
the restriction of $f_1 : M_1\to N$ to $U_1\subseteq M_1$, and analogously for $(f_2)_{U_2}^{\widetilde{\Sigma}_-}$.
Similarly, $\mu_N^{\op}\, (\FFF(f_1) \otimes \FFF(f_2))$ can be computed by
\begin{flalign}
\xymatrix@C=8em{
\FFF(M_1)\otimes\FFF(M_2) \ar[d]_-{\mathrm{flip}} \ar[r]^-{\FFF(f_1) \otimes \FFF(f_2)} &
\FFF(N)\otimes\FFF(N) \ar[r]^-{\mu_N^\op} \ar[d]_-{\mathrm{flip}} & \FFF(N) \\
\FFF(M_2)\otimes\FFF(M_1) \ar[r]^-{\FFF(f_2) \otimes \FFF(f_1)} &
\FFF(N)\otimes\FFF(N) \ar[ur]^-{\mu_N} \\
\ar[u]_-{\cong}^-{\FFF(\iota_{U_2}^{M_2}) \otimes \FFF(\iota_{U_1}^{M_1})} \FFF(U_2) \otimes \FFF(U_1) 
\ar[r]_-{\FFF\big((f_2)_{U_2}^{\widetilde{\Sigma}^\prime_+}\big) \otimes \FFF\big((f_1)_{U_1}^{\widetilde{\Sigma}^\prime_-}\big)} 
& \ar[u]_-{\cong}^-{\FFF(\iota_{\widetilde{\Sigma}^\prime_+}^N) \otimes \FFF(\iota_{\widetilde{\Sigma}^\prime_-}^N)} \FFF(\widetilde{\Sigma}^\prime_+) \otimes \FFF(\widetilde{\Sigma}^\prime_-) \ar[uur]_-{\FFF(\iota_{\und{\widetilde{\Sigma}}^\prime}^N)} 
}
\end{flalign}
The claim follows from the equivariance property \eqref{eqn:FAperm} 
of prefactorization algebras. 
\end{proof}

The main result of this subsection is as follows.
\begin{propo}\label{propo:caus}
Let $\FFF \in \PFA^{\addc}$ be any Cauchy constant additive prefactorization algebra. 
Let further $(f_1: M_1 \to N)\perp (f_2: M_2 \to N)$ be any causally disjoint pair of $\Loc$-morphisms.
Then $\mu_N^\op\, (\FFF(f_1) \otimes \FFF(f_2)) = \mu_N\, (\FFF(f_1) \otimes \FFF(f_2))$, 
where $\mu_N^{(\op)}$ denotes the (opposite) multiplication on $\FFF(N)$ from 
Proposition \ref{propo:associativeunital}. 
\end{propo}
\begin{proof}
Because $\FFF$ is by hypothesis additive (cf.\ Definition \ref{def:additivePFA}) 
and the monoidal product $\otimes$ in a cocomplete closed symmetric monoidal category 
preserves colimits in both entries, it follows that 
\begin{flalign}\label{eqn:causcolim}
\FFF(M_1)\otimes\FFF(M_2) \,\cong\, \colim_{(U_1,U_2)\in \RC_{M_1} \times \RC_{M_2}}^{~} \big(\FFF(U_1)\otimes\FFF(U_2)\big)\quad.
\end{flalign}
For every $(U_1,U_2)\in \RC_{M_1} \times \RC_{M_2}$, consider the diagram 
\begin{flalign}
\xymatrix@C=5em@R=5em{
F(U_1)\otimes\FFF(U_2) \ar@/_2.5em/[ddr]_-{\FFF\big((f_1)_{U_1}^{~}\big)\otimes\FFF\big((f_2)_{U_2}^{}\big)} \ar@/^2em/[rrd]^-{~~~~~~~\FFF\big((f_1)_{U_1}^{~}\big)\otimes\FFF\big((f_2)_{U_2}^{}\big)} \ar[dr]^-{~~~\FFF(\iota_{U_1}^{M_1})\otimes\FFF(\iota_{U_2}^{M_2})} \\
& F(M_1)\otimes\FFF(M_2) \ar[d]_-{\FFF(f_1)\otimes\FFF(f_2)} \ar[r]^-{\FFF(f_1)\otimes\FFF(f_2)} 
& F(N)\otimes\FFF(N) \ar[d]^-{\mu_N^\op} \\
& F(N)\otimes\FFF(N) \ar[r]_-{\mu_N} & \FFF(N)
}
\end{flalign}
where $(f_i)_{U_i}^{~} : U_i\to N$ denotes the restriction of $f_i:M_i\to N$
to $U_i\subseteq M_i$, for $i=1,2$.
The two triangles coincide and commute by direct inspection. Furthermore, 
for every $(U_1,U_2)\in \RC_{M_1} \times \RC_{M_2}$, the outer square commutes 
as a consequence of Lemma \ref{lem:caus} applied to the causally disjoint pair 
$(f_1)_{U_1}^{~}  \perp (f_2)_{U_2}^{~}$, 
whose images $f_1(U_1), f_2(U_2) \subseteq N$ are relatively compact subsets.
Hence, by universality of the colimit in \eqref{eqn:causcolim}, 
the inner square commutes as well, which is our claim. 
\end{proof}

Proposition \ref{propo:caus} leads to the following refinement of Corollary \ref{cor:natalg}. 
\begin{theo}\label{theo:PFAtoAQFT}
Every Cauchy constant additive prefactorization algebra $\FFF \in \PFA^{\addc}$
defines a Cauchy constant additive algebraic quantum field theory $\bbA[\FFF] \in \AQFT^{\addc}$.
Hence, the functor $\bbA : \PFA^{\addc} \to \Alg^\Loc$ from Corollary \ref{cor:natalg} 
factors through the full subcategory $\AQFT^{\addc} \subseteq \Alg^\Loc$. 
\end{theo}
\begin{proof}
Proposition \ref{propo:caus} implies that the functor $\bbA[\FFF]: \Loc \to \Alg$ 
defined in Corollary \ref{cor:natalg} is an algebraic quantum field theory, i.e.\ it satisfies 
the Einstein causality axiom \eqref{eqn:Einsteincausality}.
Because $\FFF$ is by hypothesis Cauchy constant, it follows 
that $\bbA[\FFF]$ is Cauchy constant too.
Because the underlying functors $\bbA[\FFF]\vert_M = \FFF\vert_M : \RC_M\to \CC$
to the category $\CC$ coincide, additivity of $\FFF \in \PFA^{\addc}$ and 
Remark \ref{rem:algcolim} immediately imply additivity of $\bbA[\FFF]$.
Hence, $\bbA[\FFF] \in \AQFT^{\addc}$. 
\end{proof}


\section{\label{sec:inverseconstruction}From AQFT to PFA}
In this section we show that every algebraic quantum
field theory $\AAA\in\AQFT$
defines a variant of a prefactorization algebra on $\Loc$ where the 
factorization products are defined only for those tuples of pairwise disjoint 
morphisms $\und{f} : \und{M}\to N$ that are in a suitable sense time-orderable.
We shall call this type of prefactorization algebras
{\em time-orderable} and denote the corresponding category
by $\toPFA$. Our construction defines a functor
$\bbF : \AQFT\to \toPFA$ to the category of
time-orderable prefactorization algebras. Cauchy constancy and additivity
do not play a role in this section, however we shall prove
that these properties are preserved by our functor.
\sk

Let $\AAA\in\AQFT$ be an algebraic quantum field theory.
Our aim is to construct from this data factorization products
$\bbF[\AAA](\und{f}) : \bigotimes_{i=1}^n \AAA(M_i)\to \AAA(N)$, for suitable
tuples of pairwise disjoint morphisms $\und{f} = (f_1,\dots,f_n) : \und{M}\to N$.
For $n=0$, i.e.\ the empty tuples $\emptyset \to N$, we may take
the unit $\eta_N :I \to \AAA(N)$ of the associative and unital algebra $\AAA(N)$
that is assigned by $\AAA$ to $N\in\Loc$. For $n=1$, the tuples of pairwise disjoint
morphisms are just $\Loc$-morphisms $f:M\to N$,
hence we may take the $\CC$-morphism $\bbF[\AAA](f) := \AAA(f) : \AAA(M)\to \AAA(N)$
that is obtained from the $\Alg$-morphism assigned by $\AAA$ to $f:M\to N$
via the forgetful functor $\Alg \to \CC$. For $n\geq 2$, the envisaged 
construction becomes far less obvious. Let us consider for the moment
$n=2$ and a pair of disjoint morphisms  $\und{f}=(f_1,f_2) : \und{M}\to N$.
Inspired by our previous construction \eqref{eqn:multiplicationmap} 
of multiplications from factorization products, we propose
to define $\bbF[\AAA](\und{f}) : \AAA(M_1)\otimes \AAA(M_2)\to \AAA(N)$
by the commutative diagram
\begin{flalign}\label{eqn:2facproduct}
\xymatrix@C=3em{
\ar[dr]_-{\AAA(f_1)\otimes\AAA(f_2)~~~~~}\AAA(M_1) \otimes\AAA(M_2)\ar[rr]^-{\bbF[\AAA](\und{f})} && \AAA(N)\\
&\AAA(N)\otimes \AAA(N)\ar[ru]_-{~\mu_N}&
}
\end{flalign}
in $\CC$. This is however problematic in view of the equivariance property 
\eqref{eqn:FAperm} of prefactorization algebras. In fact, if we used
\eqref{eqn:2facproduct} for {\em all} pairs of disjoint morphisms $\und{f} =(f_1,f_2):\und{M}\to N$, 
then \eqref{eqn:FAperm} would be satisfied if and only if the diagram
in \eqref{eqn:Einsteincausality} commutes, which is in general not the case
unless $f_1\perp f_2$ are causally disjoint. By closer inspection of \eqref{eqn:multiplicationmap},
one observes that \eqref{eqn:2facproduct} is {\em not} supposed to be the correct
definition for all pairs of disjoint morphisms, but only for those pairs 
$\und{f} =(f_1,f_2):\und{M}\to N$ where $f_1(M_1)\subseteq N$ is
``later'' (in a suitable sense) than $f_2(M_2)\subseteq N$. 
This would solve the problem concerning the equivariance property discussed above.
The following definition formalizes a concept of time-ordering that allows us to prove
our desired statements.
\begin{defi}\label{def:timeordered}
\begin{itemize}
\item[(a)] Let $M\in\Loc$. A tuple $(U_1,\dots,U_n)$ of causally convex open subsets $U_i\subseteq M$ 
is called {\em time-ordered} if $J^+_M(U_i) \cap U_j = \emptyset$,
for all $i<j$.

\item[(b)] A tuple of pairwise disjoint morphisms $\und{f} = (f_1,\dots,f_n) :\und{M}\to N$ is called
{\em time-ordered} if the tuple $(f_1(M_1),\ldots,f_n(M_n))$ 
of causally convex open subsets $f_i(M_i) \subseteq N$ is time-ordered.

\item[(c)] A tuple of pairwise disjoint morphisms $\und{f} = (f_1,\dots,f_n) :\und{M}\to N$ is called
{\em time-orderable} if there exists a permutation $\rho \in\Sigma_n$ such that the tuple
$\und{f}\rho = (f_{\rho(1)},\dots, f_{\rho(n)}) : \und{M}\rho\to N$ is time-ordered.
We call $\rho$ a {\em time-ordering permutation} for $\und{f}$ and note that time-ordering
permutations are not necessarily unique.
\end{itemize}
\end{defi}
\begin{rem}\label{rem:nontimeorderable}
By convention, all empty tuples $\emptyset \to N$ and all $1$-tuples
$f:M\to N$ are time-ordered. However, we would like to stress that for $n\geq 2$
not every tuple of pairwise disjoint morphisms $\und{f} : \und{M}\to N$
is time-orderable. For example, consider $n=2$ and $\und{f}= (f_1,f_2) : \und{M}\to N$
the inclusion of the following causally convex open subsets 
into the Lorentzian cylinder $N$:
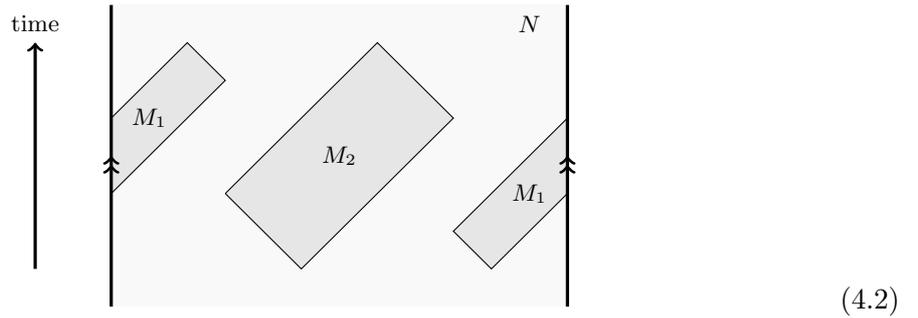
\begin{flalign}
\begin{tikzpicture}[scale=1,decoration={markings,mark=at position 0.5 with {\arrow{>>}}}]
\draw[fill=gray!5,draw=none] (-3,-2.5) -- (-3,1.5) -- (3,1.5) -- (3,-2.5) -- (3,-2.5);
\draw[fill=gray!20] (-3,0) -- (-2,1) -- (-1.5,0.5) -- (-3,-1) --(-3,0);
\draw[fill=gray!20] (3,0) -- (1.5,-1.5) -- (2,-2) -- (3,-1) -- (3,0);
\draw[very thick, postaction={decorate}] (-3,-2.5) -- (-3,1.5);
\draw[very thick, postaction={decorate}] (3,-2.5) -- (3,1.5);
\draw[fill=gray!20] (-1.5,-1) -- (0.5,1) -- (1.5,0) -- (-0.5,-2) -- (-1.5,-1);
\draw (0,-0.5) node {{\footnotesize $M_2$}};
\draw (-2.5,0) node {{\footnotesize $M_1$}};
\draw (2.5,-1) node {{\footnotesize $M_1$}};
\draw (2.5,1.25) node {{\footnotesize $N$}};
\draw[very thick, ->] (-4,-2) -- (-4,1) node[above] {{\footnotesize time}};
\end{tikzpicture}
\end{flalign}
In this picture the left and right boundaries are identified 
as indicated, thereby producing the Lorentzian cylinder 
$N=(\bbR\times \mathbb{S}^1,g=-\dd t^2 + \dd\phi^2, \mathfrak{t} = \frac{\partial}{\partial t})$.
\end{rem}

The following technical lemma is the crucial ingredient for our proofs below.
We shall use the same notation and conventions for permutation group actions as in \cite{Yau}.
\begin{lem}\label{lem:timeordered}
\begin{itemize}
\item[(i)] Let $\rho \in \Sigma_n$ be a time-ordering permutation for 
the tuple of pairwise disjoint morphisms $\und{f} =(f_1,\dots,f_n): \und{M}\to N$
and $\sigma\in\Sigma_n$ a permutation. Then $\sigma^{-1}\rho \in \Sigma_n$ 
is a time-ordering permutation for $\und{f}\sigma=(f_{\sigma(1)},\dots,f_{\sigma(n)}) 
: \und{M}\sigma \to N$.

\item[(ii)] Let $\rho_0 \in \Sigma_n$ be a time-ordering permutation for 
$\und{f} = (f_1,\dots,f_n) :\und{M}\to N$ and $\rho_i \in \Sigma_{k_i}$ a time-ordering 
permutation for  $\und{g}_i = (g_{i1},\dots, g_{ik_i}): \und{L}_i \to M_i$, for $i=1,\dots,n$. 
Then the permutation 
\begin{flalign}
\rho_0\langle k_1,\ldots,k_n\rangle\, (\rho_{\rho_0(1)}\oplus\ldots\oplus\rho_{\rho_0(n)}) \in \Sigma_{k_1+\cdots+k_n}\quad,
\end{flalign}
where $\rho_0\langle k_1,\ldots,k_n\rangle$ denotes the block permutation corresponding
to $\rho_0$ and  $\rho_{\rho_0(1)}\oplus\ldots\oplus\rho_{\rho_0(n)}$ the sum permutation
of the $\rho_{\rho_0(i)}$, is a time-ordering permutation for
\begin{flalign}\label{eqn:tuplecomposition}
\und{f}(\und{g}_1,\dots,\und{g}_n) := (f_1\,g_{11},\dots,f_n\, g_{nk_n}) \,:\, 
(\und{L}_1,\dots,\und{L}_n)~\longrightarrow~  N\quad.
\end{flalign}

\item[(iii)] Let $\und{f} : \und{M}\to N$ be a time-orderable tuple of pairwise disjoint morphisms
and $\rho,\rho^\prime\in\Sigma_n$ time-ordering permutations for $\und{f}$.
Then the right permutation $\rho^{-1}\rho^\prime : \und{f}\rho \to \und{f}\rho^\prime$
is generated by transpositions of adjacent causally disjoint pairs of morphisms.
\end{itemize}
\end{lem}
\begin{proof}
(i): Trivial. 
\sk

\noindent (ii): Since 
\begin{flalign}
\und{f}(\und{g}_1,\dots,\und{g}_n)\, 
\rho_0\langle k_1,\ldots,k_n\rangle\, (\rho_{\rho_0(1)}\oplus\ldots\oplus\rho_{\rho_0(n)}) = 
(\und{f}\rho_0)(\und{g}_{\rho_0(1)}\rho_{\rho_0(1)},\dots,\und{g}_{\rho_0(n)}\rho_{\rho_0(n)})\quad,
\end{flalign}
it is sufficient to prove that the composition of time-ordered tuples of pairwise disjoint morphisms 
is time-ordered. Therefore, assuming that $\und{f}$ and $\und{g}_i$, for $i=1,\ldots,n$, are 
time-ordered, we have to show that $(f_1\,g_{11},\dots,f_n\,g_{nk_n})$ is time-ordered,
i.e.\ $J^+_N(f_ig_{ii^\prime}(L_{ii^\prime}))\cap f_j g_{jj^\prime}(L_{jj^\prime}) =\emptyset$
for the following two cases: Case 1 is $i<j$ and arbitrary $i^\prime =1,\dots,k_i$ 
and $j^\prime = 1,\dots,k_j$. Case 2 is $i=j$ and $j<j^\prime$.
Case 1 follows immediately from the hypothesis that $\und{f}$ is time-ordered, i.e.\ 
$J^+_N(f_i(M_i)) \cap f_j(M_j) = \emptyset$ for all $i<j$. 
For case 2 we use that $\und{g}_i$ is  time-ordered,
i.e.\ $J^+_{M_i}(g_{ii^\prime}(L_{ii^\prime})) \cap g_{ij^\prime}(L_{ij^\prime}) = \emptyset$
for all $j< j^\prime$, and hence by the properties of $\Loc$-morphisms
\begin{flalign}
J^+_N(f_ig_{ii^\prime}(L_{ii^\prime})) \cap f_ig_{ij^\prime}(L_{ij^\prime}) 
= f_i\Big(J^+_{M_i}(g_{ii^\prime}(L_{ii^\prime})) \cap g_{ij^\prime}(L_{ij^\prime})\Big) = \emptyset\quad.
\end{flalign}
This proves that $(f_1\, g_{11},\ldots,f_n\, g_{nk_n})$ is time-ordered. 
\sk

\noindent (iii): Suppose that $\rho^{-1}\rho^\prime: \und{f}\rho \to \und{f}\rho^\prime$ 
reverses the time-ordering between $f_k$ and $f_\ell$, 
i.e.\ $\rho(i)=k=\rho^\prime(i^\prime)$ and $\rho(j)=\ell=\rho^\prime(j^\prime)$ 
with $i < j$ and $j^\prime < i^\prime$ or vice versa with $j < i$ and $i^\prime < j^\prime$.
Let us consider the case $i < j$ and $j^\prime < i^\prime$, the other one being similar. 
By hypothesis, we have that $J^+_N(f_{\rho(i)}(M_{\rho(i)})) \cap f_{\rho(j)}(M_{\rho(j)}) = \emptyset$ 
and $J^+_N(f_{\rho^\prime(j^\prime)}(M_{\rho^\prime(j^\prime)})) 
\cap f_{\rho^\prime(i^\prime)}(M_{\rho^\prime(i^\prime)}) = \emptyset$, 
which is equivalent to $f_k \perp f_\ell$ being causally disjoint. 
Summing up, this proves that every pair $(f_k,f_\ell)$ of morphisms
whose time-ordering is reversed by $\rho^{-1}\rho^\prime$
is causally disjoint $f_k \perp f_\ell$.
\sk

To conclude the proof, let us recall that every permutation 
$\sigma: (h_1,\dots,h_n)\to (h_{\sigma(1)},\dots,h_{\sigma(n)})$ admits a 
(not necessarily unique) factorization into adjacent transpositions 
that flip only elements whose order is reversed by $\sigma$. 
(One way to obtain such a factorization is as follows: 
Start from $(h_1,\dots,h_n)$ and move by adjacent transpositions
the element $h_{\sigma(1)}$ to the leftmost position. Then move by adjacent transpositions
the element $h_{\sigma(2)}$ to the second leftmost position, and so on.)
This implies that we obtain a factorization 
$\rho^{-1}\rho^\prime = \tau_1\cdots\tau_N: \und{f}\rho \to \und{f}\rho^\prime$, 
where each $\tau_l: \und{f}\rho\tau_1\cdots\tau_{l-1} \to \und{f}\rho\tau_1\cdots\tau_l$
transposes two adjacent $\Loc$-morphisms whose time-ordering is reversed by $\rho^{-1}\rho^\prime$.
Our result in the previous paragraph then implies that each $\tau_l$ is a transposition of adjacent causally 
disjoint pairs of morphisms, which completes our proof.
\end{proof}

Lemma~\ref{lem:timeordered} plays a crucial role in the following 
definition of time-orderable prefactorization algebras because it
ensures that time-orderable tuples of pairwise disjoint morphisms are 
composable and carry permutation actions.
A {\em time-orderable prefactorization algebra} $\FFF$ on $\Loc$ with values in $\CC$
is given by the following data:
\begin{itemize}
\item[(i)] for each $M\in \Loc$, an object $\FFF(M)\in\CC$;
\item[(ii)] for each time-orderable tuple $\und{f}=(f_1,\dots,f_n) : \und{M} \to N$
of pairwise disjoint morphisms, a $\CC$-morphism 
$\FFF(\und{f}) : \bigotimes_{i=1}^n \FFF(M_i)\to \FFF(N)$ 
(called {\em time-ordered product}),
with the convention that to the empty tuple $\emptyset \to N$ 
is assigned a morphism $I\to \FFF(N)$ from the monoidal unit.
\end{itemize}
These data are required to satisfy the analogs of the prefactorization algebra
axioms from Section \ref{subsec:PFA} for time-orderable tuples.  A morphism $\zeta: \FFF\to \GGG$ of time-orderable 
prefactorization algebras is a family $\zeta_M : \FFF(M)\to \GGG(M)$ of $\CC$-morphisms, 
for all $M\in\Loc$, that is compatible with the time-ordered products as in \eqref{eqn:PFACmorphism}.
\begin{defi}\label{def:toPFA}
We denote by $\toPFA$ the category of time-orderable prefactorization algebras on $\Loc$. 
In analogy to Definitions \ref{def:additivePFA} and \ref{def:CauchyconstantPFA},
we introduce the full subcategories $\toPFA^{\add}, \toPFA^c, \toPFA^{\addc}\subseteq \toPFA$
of additive, Cauchy constant and Cauchy constant additive time-orderable prefactorization algebras.
\end{defi}
\begin{rem}\label{rem:tPFAvsPFA}
Each ordinary prefactorization algebra on $\Loc$ defines a time-orderable one
by restriction to time-orderable tuples of pairwise disjoint morphisms. 
This defines a functor $\PFA \to \toPFA$, which is faithful, but not 
necessarily full due to the fact that  not all pairwise disjoint tuples 
$\und{f}: \und{M} \to N$ are time-orderable, cf.\ Remark~\ref{rem:nontimeorderable}. 
This functor clearly preserves both additivity and Cauchy constancy. 
\end{rem}

With these preparations we can now carry out our envisaged construction of a 
time-orderable prefactorization algebra $\bbF[\AAA] \in \toPFA$
from a given algebraic quantum field theory $\AAA \in \AQFT$. 
In particular, we can now complete our attempt from the beginning of this section 
to define the time-ordered factorization products. Let $\und{f} = (f_1,\ldots,f_n): \und{M} \to N$
be a {\em time-orderable} tuple of pairwise disjoint morphisms
with time-ordering permutation $\rho\in\Sigma_n$. We define the corresponding
time-ordered product  $\bbF[\AAA](\und{f}): \bigotimes_{i=1}^n \AAA(M_i) \to \AAA(N)$ 
by the commutative diagram 
\begin{flalign}\label{eqn:toprod}
\xymatrix@C=3em{
\ar[d]_-{\text{permute}} \bigotimes\limits_{i=1}^n \AAA(M_i) \ar[rr]^-{\bbF[\AAA](\und{f})} && \AAA(N)\\
\bigotimes\limits_{i=1}^n \AAA(M_{\rho(i)}) \ar[rr]_-{\Motimes_{i} \AAA(f_{\rho(i)})} && \AAA(N)^{\otimes n} \ar[u]_-{\mu_N^{(n)}}
}
\end{flalign}
in $\CC$, where $\mu_N^{(n)}$ denotes the $n$-ary multiplication 
in the associative and unital algebra $\AAA(N)$ in the given order, 
i.e.\ $\mu^{(n)}_N(a_1 \otimes \cdots \otimes a_n) = a_1 \cdots a_n$
with juxtaposition denoting multiplication in $\AAA(N)$.
As before, for $n=0$ we assign to the empty tuple $\emptyset \to N$ the $\CC$-morphism 
$\eta_N: I \to \AAA(N)$ corresponding to the unit of $\AAA(N)$. 
\begin{lem}\label{lem:toindependence}
The $\CC$-morphism $\bbF[\AAA](\und{f}): \bigotimes_{i=1}^n \AAA(M_i) \to \AAA(N)$ 
defined in \eqref{eqn:toprod} does not depend on the choice of time-ordering permutation 
for $\und{f}: \und{M} \to N$. 
\end{lem}
\begin{proof}
Consider time-ordering permutations $\rho, \rho^\prime \in \Sigma_n$ for $\und{f}$. 
Recalling Lemma~\ref{lem:timeordered}~(iii), the right permutation
$\rho^{-1}\rho^\prime: \und{f}\rho \to \und{f}\rho^\prime$ 
is generated by transpositions of adjacent causally disjoint pairs of morphisms. 
Hence, the claim follows from the Einstein causality axiom \eqref{eqn:Einsteincausality} 
of the algebraic quantum field theory $\AAA \in \AQFT$. 
\end{proof}

\begin{theo}\label{theo:AQFTtotPFA}
Let $\AAA \in \AQFT$ be an algebraic quantum field theory.
Then the following data defines a time-orderable prefactorization algebra $\bbF[\AAA] \in \toPFA$:
\begin{itemize}
\item[(i)] for each $M \in \Loc$, define $\bbF[\AAA](M) := \AAA(M) \in \CC$ 
via the forgetful functor $\Alg \to \CC$;

\item[(ii)] for each time-orderable tuple of pairwise disjoint morphisms 
$\und{f} = (f_1,\ldots,f_n): \und{M} \to N$, define the time-ordered product 
$\bbF[\AAA](\und{f}): \bigotimes_{i=1}^n \bbF[\AAA](M_i) \to \bbF[\AAA](N)$ 
according to \eqref{eqn:toprod} and Lemma~\ref{lem:toindependence} and, 
for each empty tuple $\emptyset \to N$, assign the unit $\eta_N: I \to \bbF[\AAA](N)$ of $\AAA(N)$. 
\end{itemize}
The assignment $\AAA \mapsto \bbF[\AAA]$ 
canonically extends to a functor $\bbF: \AQFT \to \toPFA$. 
\end{theo}
\begin{proof}
Lemma \ref{lem:timeordered} immediately implies that 
$\bbF[\AAA]$ satisfies the axioms of time-orderable prefactorization algebras.
More explicitly, Lemma \ref{lem:timeordered} (i) implies the
equivariance axiom \eqref{eqn:FAperm} for all time-orderable
tuples and Lemma \ref{lem:timeordered} (ii) implies
the composition axiom \eqref{eqn:FAcomp} for all
time-orderable tuples. By definition, we also have that
$\bbF[\AAA](\id_M) = \id_{\bbF[\AAA](M)}$, for all $M\in\Loc$.
\sk

Concerning functoriality of the assignment $\AAA \mapsto \bbF[\AAA]$,
we have to show that every $\AQFT$-morphism $\kappa :\AAA\to \BBB$ 
canonically defines a $\toPFA$-morphism $\bbF[\kappa]:\bbF[\AAA] \to \bbF[\BBB]$.
Observe that, for every time-orderable tuple $\und{f}: \und{M}\to N$ with time-ordering
permutation $\rho\in\Sigma_n$, the diagram
\begin{flalign}
\xymatrix@C=5em{
\bigotimes\limits_{i=1}^n \AAA(M_i) \ar[r]^-{\text{permute}} \ar[d]_-{\Motimes_{i} \kappa_{M_i}} 
& \bigotimes\limits_{i=1}^n \AAA(M_{\rho(i)}) \ar[r]^-{\Motimes_{i} \AAA(f_{\rho(i)})} \ar[d]_-{\Motimes_{i}\kappa_{M_{\rho(i)}}}
& \AAA(N)^{\otimes n} \ar[r]^-{\mu^{(n)}_N} \ar[d]_-{\kappa_N^{\otimes n}} 
& \AAA(N) \ar[d]^-{\kappa_N} \\ 
\bigotimes\limits_{i=1}^n \BBB(M_i) \ar[r]_-{\text{permute}} 
& \bigotimes\limits_{i=1}^n \BBB(M_{\rho(i)}) \ar[r]_-{\Motimes_{i} \BBB(f_{\rho(i)})} 
& \BBB(N)^{\otimes n} \ar[r]_-{\mu^{(n)}_N} & \BBB(N)
}
\end{flalign}
in $\CC$ commutes. Hence, the family $\kappa_M: \AAA(M) \to \BBB(M)$ of $\CC$-morphisms 
defines a $\toPFA$-morphism $\bbF[\kappa]: \bbF[\AAA] \to \bbF[\BBB]$. 
\end{proof}

\begin{propo}\label{propo:AQFTtotPFAaddc}
$\AAA \in \AQFT$ is additive (respectively Cauchy constant) 
if and only if $\bbF[\AAA] \in \toPFA$ is additive (respectively Cauchy constant).
In particular, the functor $\bbF: \AQFT \to \toPFA$ from Theorem \ref{theo:AQFTtotPFA}
restricts to full subcategories as
$\bbF: \AQFT^{\add} \to \toPFA^{\add}$, $\bbF: \AQFT^{c} \to \toPFA^{c}$ 
and $\bbF : \AQFT^{\addc} \to \toPFA^{\addc}$. 
\end{propo}
\begin{proof}
Let us recall that, by our construction, the underlying functors 
$\bbF[\AAA] = \AAA : \Loc\to \CC$ to the category $\CC$ coincide.
It is then a consequence of Remark \ref{rem:algcolim}
that $\bbF[\AAA]$ is additive if and only if  $\AAA$ is additive.
Furthermore, because the forgetful functor $\Alg \to \CC$ preserves 
and detects isomorphisms, it follows that $\bbF[\AAA]$ is Cauchy constant
if and only if $\AAA$ is Cauchy constant.
\end{proof}


\section{\label{sec:equivalence}Equivalence theorem}

\subsection{Main result}
The aim of this section is to prove that our two constructions
from Sections \ref{sec:construction} and \ref{sec:inverseconstruction} 
are inverse to each other when restricted to their common domain of validity.
Recall that in Section \ref{sec:construction} we considered
Cauchy constant additive prefactorization algebras
and constructed a functor $\bbA : \PFA^{\addc} \to \AQFT^{\addc}$
to the category of Cauchy constant additive algebraic quantum field theories,
cf.\ Theorem \ref{theo:PFAtoAQFT}. Because the construction 
presented in Section \ref{sec:construction} only involves 
{\em time-orderable} tuples of disjoint morphisms,
this functor factors through the forgetful functor
$\PFA^{\addc} \to \toPFA^{\addc}$ (cf.\ Remark \ref{rem:tPFAvsPFA}) to the 
category of Cauchy constant additive {\em time-orderable}
prefactorization algebras, cf.\ Definition \ref{def:toPFA}.
We shall denote the resulting functor by the same symbol
$\bbA : \toPFA^{\addc} \to\AQFT^{\addc}$.
Let us further recall the functor $\bbF : \AQFT^{\addc} \to \toPFA^{\addc}$ from 
Theorem \ref{theo:AQFTtotPFA} and Proposition \ref{propo:AQFTtotPFAaddc}.
Our main result is the following equivalence theorem.
\begin{theo}\label{theo:equivalence}
The two functors $\bbA : \toPFA^{\addc} \to \AQFT^{\addc}$ and 
$\bbF : \AQFT^{\addc} \to \toPFA^{\addc}$ 
are inverses of each other. As a consequence, the
category $\AQFT^{\addc}$ of Cauchy constant additive 
algebraic quantum field theories is isomorphic to the category
$\toPFA^{\addc}$ of Cauchy constant additive time-orderable prefactorization algebras.
\end{theo}
\begin{proof}
The only non-trivial check to confirm that
$\bbA \circ \bbF = \id_{\AQFT^{\addc}}$
amounts to show that, for every $\AAA \in \AQFT^{\addc}$, 
the multiplications on $\bbA[\bbF[\AAA]](M)$ 
and on $\AAA(M)$ coincide, for all $M\in\Loc$. 
By \eqref{eqn:multiplicationmap} and \eqref{eqn:toprod}, 
the multiplication on $\bbA[\bbF[\AAA]](M)$  is given by 
\begin{flalign}
\xymatrix@C=2.5em{
\AAA(M) \otimes \AAA(M) && \ar[ll]_-{\AAA(\iota_{U_+}^M) \otimes \AAA(\iota_{U_-}^M)}^-{\cong} 
\AAA(U_+) \otimes \AAA(U_-) \ar[rr]^-{\AAA(\iota_{U_+}^M) \otimes \AAA(\iota_{U_-}^M)} 
&& \AAA(M)^{\otimes 2} \ar[r]^-{\mu_M} & \AAA(M)
}\quad,
\end{flalign}
where $\iota_{\und{U}}^M = (\iota_{U_+}^M,\iota_{U_-}^M): \und{U}\to M$ is any object of $\PP_M$. 
This clearly coincides with the original multiplication $\mu_M$ on $\AAA(M)$. 
\sk

Conversely, to show that $\bbF \circ \bbA = \id_{\toPFA^{\addc}}$, 
we have to confirm that the time-ordered products of $\bbF[\bbA[\FFF]] \in \toPFA^{\addc}$ 
coincide with the original time-ordered products of $\FFF \in \toPFA^{\addc}$. 
In arity $n=0$ and $n=1$ this is obvious. For $n\geq 2$, this is more complicated and requires some
preparations. Using equivariance under permutation actions, it is sufficient to compare 
the time-ordered products for {\em time-ordered} (in contrast to time-orderable) 
tuples $\und{f}=(f_1,\ldots,f_n): \und{M} \to N$. Because of additivity, we can further
restrict to the case where $\und{f}: \und{M} \to N$ has relatively compact images, i.e.\ 
$f_i(M_i)\subseteq N$ is relatively compact, for all $i=1,\dots,n$.
We shall now show that, due to Cauchy constancy, we can further restrict our attention 
to time-ordered tuples $\und{h}= (h_1,\dots,h_n): \und{L} \to N$ with relatively compact 
images for which there exists a Cauchy surface $\Sigma$ of $N$ such that
\begin{flalign}\label{eqn:assumption}
h_1(L_1),\dots, h_{n-1}(L_{n-1}) \subseteq \Sigma_+ := I^+_N(\Sigma) \subseteq N 
\qquad\mbox{and}\qquad
h_{n}(L_n) \subseteq \Sigma_- := I^-_N(\Sigma) \subseteq N\quad.
\end{flalign}
Indeed, given any time-ordered tuple $\und{f}: \und{M} \to N$ with relatively compact 
images, we shall prove below that there exists a family of Cauchy morphisms
$g_i : L_i \stackrel{c}{\to} M_i$, for $i=1,\dots,n$, such that
$\und{h}:= \und{f}(g_1,\dots,g_n)= (f_1\,g_1,\dots,f_n\, g_n) : \und{L} \to N$ 
admits a Cauchy surface $\Sigma$ that satisfies \eqref{eqn:assumption}. Cauchy constancy 
and the fact that the time-ordered products of $\bbF[\bbA[\FFF]]$
and $\FFF$ agree in arity $n=1$ then implies that $\bbF[\bbA[\FFF]](\und{f}) = \FFF(\und{f})$
if and only if $\bbF[\bbA[\FFF]](\und{h}) = \FFF(\und{h})$. To exhibit such
a family of Cauchy morphisms for $\und{f}: \und{M} \to N$, 
let us choose Cauchy surfaces $\Sigma_i$ of $M_i$, for 
$i=1,\dots,n$, and define $L_i := I^+_{M_i}(\Sigma_i)$, for $i=1,\dots,n-1$,
and $L_n := I^-_{M_n}(\Sigma_n)$. Let us further define
$g_i := \iota_{L_i}^{M_i} : L_i\stackrel{c}{\to} M_i$ by subset inclusion, for $i=1,\dots,n$.
A Cauchy surface $\Sigma$ of $N$ is constructed
by extending via \cite[Theorem~3.8]{BernalSanchez} the 
compact and achronal subset
\begin{flalign}
\widetilde{\Sigma} \,:=\, \bigcup_{i=1}^n \bigg( \overline{f_i(\Sigma_i)} \Big\backslash I^+_N \Big( \bigcup_{j=i+1}^n \overline{f_j(\Sigma_j)} \Big) \bigg) \subseteq N \quad.
\end{flalign}
By direct inspection one observes that $\Sigma$ fulfills \eqref{eqn:assumption}. 
\sk

Using \eqref{eqn:assumption}, we obtain a factorization
\begin{flalign}
\und{h} = \iota_{\und{\Sigma}}^N \big((h_1^{\Sigma_+},\dots,h_{n-1}^{\Sigma_+}), h_n^{\Sigma_-}\big)\quad,
\end{flalign}
where on the right-hand side we regard $h_i^{\Sigma_+} : L_i \to \Sigma_+$
as morphisms to $\Sigma_+$, for $i=1,\dots,n-1$, and
$h_n^{\Sigma_-} : L_n\to \Sigma_-$ as a morphism to $\Sigma_-$.
Iterating this construction, we observe that it is sufficient
to prove that $\bbF[\bbA[\FFF]](\iota_{\und{\Sigma}}^N) = \FFF(\iota_{\und{\Sigma}}^N)$,
for all $\iota_{\und{\Sigma}}^N = (\iota_{\Sigma_+}^N,\iota_{\Sigma_-}^N) : \und{\Sigma}\to N$,
where $N\in\Loc$ and the Cauchy surface $\Sigma$ of $N$ is arbitrary.
Using \eqref{eqn:toprod} and \eqref{eqn:multiplicationmap},
we obtain that $\bbF[\bbA[\FFF]](\iota_{\und{\Sigma}}^N ) : \FFF(\Sigma_+)\otimes\FFF(\Sigma_-)\to \FFF(N)$
is given by
\begin{flalign}
\xymatrix@C=2.5em{
\FFF(\Sigma_+) \otimes \FFF(\Sigma_-) \ar[rr]^-{\FFF(\iota_{\Sigma_+}^N)\otimes\FFF(\iota_{\Sigma_-}^N)} && 
\FFF(N)\otimes \FFF(N) && \ar[ll]_-{\FFF(\iota_{\Sigma_+}^N)\otimes\FFF(\iota_{\Sigma_-}^N)}^-{\cong} 
\FFF(\Sigma_+) \otimes \FFF(\Sigma_-)  \ar[r]^-{\FFF(\iota_{\und{\Sigma}}^N)}& \FFF(N)
}\quad,
\end{flalign}
which clearly coincides with the original time-ordered product $\FFF(\iota_{\und{\Sigma}}^N):
\FFF(\Sigma_+)\otimes \FFF(\Sigma_-)\to \FFF(N)$. This concludes our proof.
\end{proof}

\begin{rem}\label{rem:operad}
We would like to mention very briefly a more abstract operadic perspective on the
Equivalence Theorem \ref{theo:equivalence}. Recall from \cite{BSWoperad}
that there exists a $\Set$-valued colored operad $\O_{(\Loc,\perp)}$
whose category of $\CC$-valued algebras is the category
of algebraic quantum field theories, i.e.\ $\AQFT = \Alg_{\O_{(\Loc,\perp)}}(\CC)$. 
We can also define a $\Set$-valued colored operad $\P_{\Loc}$ such that
$\toPFA = \Alg_{\P_{\Loc}}(\CC)$. Concretely, the colors of
$\P_{\Loc}$ are the objects of $\Loc$ and the sets of operations
are $\P_{\Loc}\big(\substack{N \\ \und{M}}\big) := \big\{\text{all time-orderable tuples  }\und{f} : \und{M}\to N\big\}$.
Operadic composition is given by \eqref{eqn:tuplecomposition}, the operadic units
are $\id_M \in \P_{\Loc}\big(\substack{M\\M}\big)$ and the permutation actions
are $\P_{\Loc}(\sigma) : \P_{\Loc}\big(\substack{N \\ \und{M}}\big)\to 
\P_{\Loc}\big(\substack{N \\ \und{M}\sigma}\big)\,,~\und{f}\mapsto \und{f}\sigma$, 
for $\sigma\in\Sigma_n$. Using Lemma \ref{lem:timeordered} and the definition
of the colored operad $\O_{(\Loc,\perp)}$ given in \cite{BSWoperad}, one immediately
observes that the component maps
\begin{flalign}
\Phi \, :\,  \P_{\Loc}\big(\substack{N\\\und{M}}\big)~\longrightarrow~\O_{(\Loc,\perp)}\big(\substack{N\\\und{M}}\big)
 ~,~~ \und{f}\, \longmapsto \,\big[\rho^{-1},\und{f}\big]
\end{flalign} 
define a colored operad morphism $\Phi : \P_{\Loc}\to \O_{(\Loc,\perp)}$, 
where $\rho\in\Sigma_n$ is any time-ordering permutation for $\und{f}$. The associated
pullback functor $\Phi^\ast : \Alg_{\O_{(\Loc,\perp)}}(\CC)\to \Alg_{\P_{\Loc}}(\CC)$
is then precisely our functor $\bbF : \AQFT\to \toPFA$  from Theorem \ref{theo:AQFTtotPFA}.
By operadic left Kan extension, there exists an adjunction
\begin{flalign}\label{eqn:adjunction}
\xymatrix{
\Phi_! \,:\, \toPFA ~\ar@<0.5ex>[r] & \ar@<0.5ex>[l]~ \AQFT \,:\,  \Phi^\ast= \bbF
}\quad.
\end{flalign}
Theorem \ref{theo:equivalence} then states that restricting both sides of this adjunction to Cauchy constant 
and additive theories induces an adjoint equivalence $\bbA: \toPFA^{\addc} \stackrel{\sim}{\rightleftarrows}
\AQFT^{\addc} : \bbF$.
\sk

We expect that this operadic perspective will become important when considering
the case where the target category $\CC$ is a higher category or model category.
This generalization is crucial for the description of quantum gauge theories in terms of factorization algebras
\cite{CostelloGwilliam} or algebraic quantum field theories \cite{BSShocolim,BSfibered,BSWhomotopy,BSoverview}.
The adjunction \eqref{eqn:adjunction} then becomes a Quillen adjunction
between model categories, and a reasonable equivalence theorem would
state that suitable restrictions 
to homotopy-invariant analogs of Cauchy constant and additive theories
induce a Quillen equivalence. Proving such an equivalence theorem
in a higher categorical context is technically complicated and 
will not be considered in the present paper.
\end{rem}

\subsection{\label{subsec:Involutions}Transfer of $\ast$-involutions}
Algebraic quantum field theories are typically endowed with the structure
of a $\ast$-involution, i.e.\ they assign $\ast$-algebras to spacetimes.
The aim of this subsection is to introduce $\ast$-involutions
for Cauchy constant additive time-orderable prefactorization algebras
by transferring via our Equivalence Theorem \ref{theo:equivalence} 
the usual concept of $\ast$-involution for algebraic quantum field theories.
The formalization of $\ast$-structures requires the 
underlying category $\CC$ to be an {\em involutive category}, 
see e.g.\  \cite{BSWinvolutivecats}. To simplify our presentation,
we consider only the most relevant case where $\CC=\Vec_{\bbC}^{}$
is the symmetric monoidal category of complex vector spaces,
endowed with the usual involution functor $\overline{(-)} : \Vec_\bbC^{}\to\Vec_{\bbC}^{}$
that assigns to a complex vector space $V\in\Vec_{\bbC}^{}$ its complex 
conjugate vector space $\overline{V}\in\Vec_{\bbC}^{}$. The complex conjugate
of a $\bbC$-linear map $L : V\to W$ is denoted by $\overline{L} : \overline{V}\to \overline{W}$.
We note that $\overline{\overline{V}} = V$, for all $V\in\Vec_{\bbC}^{}$,
and that $\overline{V}\otimes\overline{W} = \overline{V\otimes W}$, for all $V,W\in\Vec_{\bbC}^{}$.
Moreover, complex conjugation on $\bbC$ defines a $\bbC$-linear map
$\ast : \bbC\to\overline{\bbC}$ that satisfies $\overline{\ast} \circ \ast = \id_{\bbC} : \bbC\to
\overline{\overline{\bbC}}=\bbC$.
\sk

The results in \cite{BSWinvolutivecats} allow us to endow the category 
$\AQFT$ of $\Vec_{\bbC}^{}$-valued algebraic quantum field theories 
with an involutive structure, which we denote with an abuse of notation
also by $\overline{(-)} : \AQFT\to\AQFT$. Concretely, the complex conjugate
$\overline{\AAA}\in \AQFT$ of $\AAA\in\AQFT$ is determined by the functor
$\overline{\AAA} : \Loc\to \Alg$ that assigns to $M\in\Loc$
the algebra $\overline{\AAA}(M)$ whose underlying vector space is 
$\overline{\AAA(M)}$ and whose multiplication and unit are
$\overline{\mu_{M}^{\op}} :\overline{\AAA(M)}\otimes \overline{\AAA(M)} = 
\overline{\AAA(M)\otimes\AAA(M)}\to \overline{\AAA(M)}$
and $\overline{\eta_M}\circ \ast : \bbC\to \overline{\bbC} \to \overline{\AAA(M)}$.
(The opposite multiplication appears here because the relevant 
$\ast$-involutions on algebras are order-reversing, i.e.\ $(a\,b)^\ast = b^\ast\,a^\ast$.)
To a $\Loc$-morphism $f:M\to N$, it assigns the algebra morphism
determined by the complex conjugate $\bbC$-linear map
$\overline{\AAA}(f) := \overline{\AAA(f)} : \overline{\AAA(M)}\to\overline{\AAA(N)}$. 
We note that $\overline{\overline{\AAA}} = \AAA$, for all $\AAA\in\AQFT$.
A {\em $\ast$-involution} on an algebraic quantum field theory $\AAA\in\AQFT$
is then defined as an $\AQFT$-morphism $\ast_{\AAA} : \AAA\to \overline{\AAA} $
that satisfies $\overline{\ast_{\AAA}}\circ \ast_{\AAA} =\id_{\AAA} : 
\AAA\to \overline{\overline{\AAA}}=\AAA $. We denote by $\ast\AQFT$
the category whose objects are pairs $(\AAA,\ast_\AAA)$ consisting 
of an $\AAA\in\AQFT$ and a $\ast$-involution $\ast_\AAA$ 
and whose morphisms are $\AQFT$-morphisms $\kappa :\AAA\to \BBB$ 
that preserve the $\ast$-involutions, i.e.\ $\overline{\kappa} \circ \ast_{\AAA} = \ast_\BBB\circ \kappa$.
It is easy to confirm that our definition agrees with the usual one from the literature 
\cite{Brunetti,FewsterVerch,AQFTbook} that considers
functors  $\Loc\to \ast\Alg$ to the category of $\ast$-algebras over $\bbC$,
see \cite{BSWinvolutivecats} for more details.
\sk

The involutive structure on $\AQFT$ restricts
to an involution functor $\overline{(-)} : \AQFT^{\addc}\to \AQFT^{\addc}$
on the full subcategory of Cauchy constant additive algebraic quantum field theories.
By the Equivalence Theorem  \ref{theo:equivalence}, we obtain
a transferred involution functor $\overline{(-)} : \toPFA^{\addc}\to\toPFA^{\addc}$
on the category of Cauchy constant additive time-orderable prefactorization algebras,
which we denote with an abuse of notation by the same symbol.
Concretely, the complex conjugate $\overline{\FFF}\in \toPFA^{\addc}$
of $\FFF \in \toPFA^{\addc}$ is given by $\overline{\FFF} := \bbF[\overline{\bbA[\FFF]}]$.
A {\em $\ast$-involution} on a Cauchy constant additive time-orderable prefactorization
algebra $\FFF \in \toPFA^{\addc}$ is then defined as a $\toPFA^{\addc}$-morphism
$\ast_\FFF : \FFF\to\overline{\FFF}$ that satisfies 
$\overline{\ast_\FFF}\circ\ast_{\FFF} = \id_{\FFF} : \FFF\to\overline{\overline{\FFF}}=\FFF$.
We denote by $\ast\toPFA^{\addc}$ the category whose objects are pairs $(\FFF,\ast_\FFF)$ 
consisting of a $\FFF\in\toPFA^{\addc}$ and a $\ast$-involution $\ast_\FFF$ 
and whose morphisms are $\toPFA^{\addc}$-morphisms $\zeta :\FFF\to \GGG$ 
that preserve the $\ast$-involutions, i.e.\ $\overline{\zeta} \circ \ast_{\FFF} = \ast_\GGG\circ \zeta$.
By construction, the Equivalence Theorem \ref{theo:equivalence} determines
an equivalence $\ast\AQFT^{\addc}\simeq \ast\toPFA^{\addc}$ 
between theories with $\ast$-involutions.
\sk 

From our constructions above, it remains unclear if there exists an {\em intrinsic}
definition of the complex conjugate prefactorization algebra 
$\overline{\FFF} = \bbF[\overline{\bbA[\FFF]}] \in\toPFA^{\addc}$
that does not rely on Cauchy constancy and additivity, i.e.\
that is applicable to {\em all} time-orderable prefactorization algebras in $\toPFA$.
Unfortunately, this does not seem to be the case. To understand
and explain these issues, let us compute explicitly the complex conjugate 
factorization product $\overline{\FFF}(\iota_{(\Sigma_+,\Sigma_-)}^M) : 
\overline{\FFF}(\Sigma_+)\otimes \overline{\FFF}(\Sigma_-)\to\overline{\FFF}(M)$
for the {\em time-ordered} pair of inclusions $\iota_{\Sigma_\pm}^M: \Sigma_\pm \to M$
determined by a choice of Cauchy surface $\Sigma\subset M$ via $\Sigma_\pm = I^\pm_M(\Sigma)$.
Using \eqref{eqn:toprod} and \eqref{eqn:multiplicationmap}, we obtain
the commutative diagram
\begin{flalign}\label{eqn:involutiondiagram}
\xymatrix@C=7em{
\ar@{=}[d]\overline{\FFF(\Sigma_+)}\otimes \overline{\FFF(\Sigma_-)} \ar[rr]^-{\overline{\FFF}(\iota_{(\Sigma_+,\Sigma_-)}^M)} && \overline{\FFF(M)}\\
\overline{\FFF(\Sigma_+)\otimes \FFF(\Sigma_-)} \ar[r]_-{\overline{\FFF(\iota_{\Sigma_+}^M) \otimes \FFF(\iota_{\Sigma_-}^M)}} & \overline{\FFF(M)\otimes\FFF(M)} & \ar[l]_-{\cong}^-{\overline{\FFF(\iota_{\Sigma_-}^M)\otimes \FFF(\iota_{\Sigma_+}^M)}} \overline{\FFF(\Sigma_-)\otimes\FFF(\Sigma_+)}\ar[u]_-{\overline{\FFF(\iota_{(\Sigma_-,\Sigma_+)}^M)}}
}
\end{flalign}
which relates the factorization product $\overline{\FFF}(\iota_{(\Sigma_+,\Sigma_-)}^M)$
of $\overline{\FFF}$ to the factorization product $\FFF(\iota_{(\Sigma_-,\Sigma_+)}^M)$
of $\FFF$. Note that the bottom horizontal arrow uses Cauchy constancy explicitly.
Physically speaking, it propagates observables from the future region $\Sigma_+$ to the 
past region $\Sigma_-$ and observables from $\Sigma_-$ to $\Sigma_+$. 
In particular, in absence of Cauchy constancy,
the diagram in \eqref{eqn:involutiondiagram} can not be used to determine
the factorization product $\overline{\FFF}(\iota_{(\Sigma_+,\Sigma_-)}^M)$
from the factorization products of $\FFF$, because the second bottom horizontal 
arrow is in general not invertible.

\subsection{\label{subsec:KleinGordon}Example: The free Klein-Gordon field}
We apply our general Equivalence Theorem \ref{theo:equivalence}
to the simple example given by the free Klein-Gordon field and thereby 
recover the results from \cite{GwilliamRejzner}.
Let us briefly recall the algebraic quantum field theory
description of the free Klein-Gordon field. For every $M\in\Loc$,
consider the Klein-Gordon operator $P_M := -\square_M + m^2 : C^\infty(M)\to C^\infty(M)$,
where $\square_M$ is the d'Alembert operator and $m^2\geq 0$ is a mass parameter.
$P_M$ admits a unique retarded/advanced Green's operator $G^\pm_M : C^\infty_\cc(M)\to C^\infty(M)$, 
where the subscript `$\cc$' denotes compactly supported functions. The $\bbR$-vector space
$\mathcal{V}(M)$ of linear observables on $M$ is defined as the cokernel
\begin{flalign}
\xymatrix{
C^\infty_\cc(M) \ar[r]^-{P_M} & C^\infty_\cc(M) \ar[r]& \mathcal{V}(M):=C^\infty_\cc(M)\big/ P_M(C^\infty_\cc(M))
}\quad.
\end{flalign}
Because $C^\infty_\cc : \Loc\to \Vec_\bbR^{}$ is a cosheaf for (causally convex) open covers
and $P : C^\infty_\cc\to C^\infty_\cc$ is a natural transformation, it follows that
$\mathcal{V} : \Loc\to\Vec_\bbR^{}$ is a cosheaf too.
Consider the complexified symmetric algebra 
$\mathrm{Sym}_\bbC (\mathcal{V}(M))\in\mathbf{CAlg}$, which is a commutative algebra in the
closed symmetric monoidal category $(\Vec_\bbC^{},\otimes,\bbC,\tau)$ of complex vector spaces.
This algebra is deformed to a noncommutative algebra by introducing a $\star$-product.
For this we first define a (de Rham type) differential $\dd : \mathrm{Sym}_\bbC (\mathcal{V}(M))\to
\mathrm{Sym}_\bbC (\mathcal{V}(M))\otimes \mathcal{V}(M)$ by setting on monomials
\begin{flalign}
\dd\big(\varphi_1\,\cdots\, \varphi_n\big) := \sum_{i=1}^n \varphi_1\,\cdots \,\omi{i} \,\cdots\,\varphi_n \otimes\varphi_i\quad,
\end{flalign}
where $\omi{i}$ means omission of $\varphi_i$. Using the causal propagator
$G_M := G^+_M-G^-_M : \mathcal{V}(M) \to \ker P_M$ and the integration map
$\int_M :  \mathcal{V}(M)\otimes \ker P_M \to\bbR\,,~\varphi\otimes \Phi \mapsto \int_M \varphi\,\Phi\,\vol_M$,
we define  the bi-differential operator
\begin{flalign}
\xymatrix@C=2.5em{
\ar[d]_-{(\id\otimes\tau\otimes\id)\circ (\dd\otimes\dd) }\mathrm{Sym}_\bbC (\mathcal{V}(M))^{\otimes 2} \ar[rr]^-{\langle G_M,\dd\otimes\dd\rangle} && \mathrm{Sym}_\bbC (\mathcal{V}(M))^{\otimes 2}\\
\mathrm{Sym}_\bbC (\mathcal{V}(M))^{\otimes 2}\otimes \mathcal{V}(M)\otimes\mathcal{V}(M)\ar[rr]_-{\id\otimes \id\otimes G_M}&&\mathrm{Sym}_\bbC (\mathcal{V}(M))^{\otimes 2}\otimes\mathcal{V}(M)\otimes\ker P_M \ar[u]_-{\id\otimes \int_M}
}
\end{flalign}
where we recall that $\tau$ is the symmetric braiding on $\Vec_\bbC^{}$, i.e.\ the flip map.
The $\star$-product $\star_M : \mathrm{Sym}_\bbC (\mathcal{V}(M))^{\otimes 2}\to 
\mathrm{Sym}_\bbC (\mathcal{V}(M))$ is then defined by composing
\begin{flalign}\label{eqn:starproduct}
\xymatrix@C=4em{
\mathrm{Sym}_\bbC (\mathcal{V}(M))^{\otimes 2}\ar[rr]^-{\exp\big(\frac{i}{2}\,\langle G_M,\dd\otimes\dd\rangle\big)}
&& \mathrm{Sym}_\bbC (\mathcal{V}(M))^{\otimes 2} \ar[r]^-{\cdot_M}& \mathrm{Sym}_\bbC (\mathcal{V}(M))
}\quad,
\end{flalign}
where $\cdot_M$ denotes the commutative product on $\mathrm{Sym}_\bbC (\mathcal{V}(M))$.
(The exponential series converges because it terminates for polynomials.) 
Setting $\AAA_{\mathrm{KG}}(M) := \big( \mathrm{Sym}_\bbC (\mathcal{V}(M)),\star_M,\eta_M\big)\in\Alg$
with $\eta_M$ the unit of  $\mathrm{Sym}_\bbC (\mathcal{V}(M))$, for all $M\in\Loc$, 
defines a Cauchy constant additive algebraic quantum field theory $\AAA_{\mathrm{KG}}\in\AQFT^{\addc}$.
Note that additivity is a consequence of $\mathcal{V}:\Loc\to\Vec_\bbR^{}$ being a cosheaf.
\sk

Theorem \ref{theo:equivalence} provides a corresponding
Cauchy constant additive time-orderable prefactorization algebra
$\FFF_{\mathrm{KG}} := \bbF[\AAA_{\mathrm{KG}}] \in \toPFA^{\addc}$.
To get some intuition on what this prefactorization algebra does,
let us analyze the explicit form of the binary time-ordered products
$\FFF_{\mathrm{KG}} (\und{f}) : \FFF_{\mathrm{KG}}(M_1)\otimes\FFF_{\mathrm{KG}}(M_2)\to 
\FFF_{\mathrm{KG}}(N)$. In the case where $\und{f}=(f_1,f_2) :\und{M}\to N$ is time-ordered,
i.e.\ $J^+_N(f_1(M_1))\cap f_2(M_2)=\emptyset$, we obtain from \eqref{eqn:toprod}, 
\eqref{eqn:starproduct} and the support properties of $G^\pm_N$ that
\begin{subequations}\label{eqn:timeantitime}
\begin{flalign}
\FFF_{\mathrm{KG}} (\und{f}) = \cdot_N \circ \exp\big(\tfrac{i}{2}\,\langle G^+_N,\dd\otimes\dd\rangle\big)\circ \big( \AAA_{\mathrm{KG}}(f_1)\otimes \AAA_{\mathrm{KG}}(f_2)\big)\quad\text{($\und{f}$ time-ordered)}\quad.
\end{flalign}
In the case where $\und{f}=(f_1,f_2) : \und{M}\to N$ is anti-time-ordered,
i.e.\ $J^+_N(f_2(M_2))\cap f_1(M_1)=\emptyset$, we obtain
\begin{flalign}
\FFF_{\mathrm{KG}} (\und{f}) = \cdot_N \circ \exp\big(\tfrac{i}{2}\,\langle G^-_N,\dd\otimes\dd\rangle\big)\circ \big( \AAA_{\mathrm{KG}}(f_1)\otimes \AAA_{\mathrm{KG}}(f_2)\big)\quad\text{($\und{f}$ anti-time-ordered)}\quad.
\end{flalign}
\end{subequations}
Using again the support properties of $G^\pm_N$, we observe that the
the two cases in \eqref{eqn:timeantitime} can be combined
into a single formula
\begin{flalign}
\FFF_{\mathrm{KG}} (\und{f}) = \cdot_N \circ \exp\big(i\,\langle G^{\mathrm{D}}_N,\dd\otimes\dd\rangle\big)\circ \big( \AAA_{\mathrm{KG}}(f_1)\otimes \AAA_{\mathrm{KG}}(f_2)\big)\quad\text{($\und{f}$ time-orderable)}\quad,
\end{flalign}
where $G^{\mathrm{D}}_N := \tfrac{1}{2} (G^+_N + G^-_N)$ is the so-called Dirac propagator,
that is valid for every time-orderable tuple $(f_1,f_2)$. In perturbative algebraic
quantum field theory (see e.g.\ \cite{FredenhagenRejzner, Rejzner}), the products $\cdot_{\mathcal{T}_N}^{} := 
\cdot_N \circ \exp\big(i\,\langle G^{\mathrm{D}}_N,\dd\otimes\dd\rangle\big)$ are called
time-ordered products. 
\sk

Our observations in this subsection can thus be summarized as follows:
The prefactorization algebra $\FFF_{\mathrm{KG}}\in\toPFA^{\addc}$ corresponding
to the free Klein-Gordon theory $\AAA_{\mathrm{KG}}\in\AQFT^{\addc}$ 
encodes the usual time-ordered products obtained by the Dirac propagator. This 
agrees with the observations in \cite{GwilliamRejzner}.



\section*{Acknowledgments}
We would like to thank the anonymous referees
for useful comments that helped us to improve this manuscript.
The work of M.B.\ is supported by a research grant funded by 
the Deutsche Forschungsgemeinschaft (DFG, Germany). 
M.P.\ is supported by a PhD scholarship of the Royal Society (UK).
A.S.\ gratefully acknowledges the financial support of 
the Royal Society (UK) through a Royal Society University 
Research Fellowship, a Research Grant and an Enhancement Award.


\end{document}